\documentclass[12pt, draftclsnofoot, onecolumn]{IEEEtran}
\usepackage{scalerel}
\usepackage{enumitem}
\usepackage{tablefootnote}
\usepackage{xcolor}
\usepackage{soul}
\usepackage[sumlimits,intlimits]{amsmath}
\usepackage{amsfonts,amssymb}%
\usepackage{amsthm}
\usepackage{bm}%
\usepackage{graphicx,graphics}%
\usepackage{cases}%
\usepackage[noadjust]{cite}%
\usepackage{color}%
\usepackage{cite,url}
\usepackage{verbatim}
\usepackage{algorithm}
\usepackage{algorithmic}
\usepackage{balance}
\usepackage{flushend}
\usepackage{multirow}
\usepackage{stfloats}
\usepackage{subfigure}
\usepackage{xspace}
\usepackage{smartref}
\usepackage{epstopdf}
\usepackage{comment}
\usepackage[intlimits]{amsmath}
\usepackage{setspace}
\usepackage{array}
\usepackage[affil-it]{authblk}
\usepackage{arydshln}
\usepackage{pgfplots}
\pgfplotsset{compat=newest}
\pgfplotsset{plot coordinates/math parser=false}
\usepackage{tikz}
\usetikzlibrary{external}
\DeclareMathOperator{\diag}{diag}
\DeclareMathOperator{\CRLB}{CRLB}
\DeclareMathOperator{\Tr}{Tr}
\DeclareMathOperator{\sinc}{sinc}
\newtheorem{definition}{Definition}
\newtheorem{proposition}{Proposition}
\setlist{leftmargin=4mm}
\allowdisplaybreaks
\begin{document}
\title{{Error Bounds for Uplink and Downlink 3D Localization in 5G mmWave Systems}}
\author{Zohair Abu-Shaban,~\IEEEmembership{Member, IEEE}, Xiangyun Zhou,~\IEEEmembership{Member, IEEE},  Thushara Abhayapala, ~\IEEEmembership{Senior Member, IEEE},  Gonzalo Seco-Granados,~\IEEEmembership{Senior Member, IEEE},  and Henk Wymeersch,~\IEEEmembership{Member, IEEE},
\vspace{-18mm}
}
\maketitle
\thispagestyle{empty}
{{\let\thefootnote\relax\footnotetext{
Zohair Abu-Shaban was with the Research School of Engineering (RSEng), the Australian National University (ANU), Canberra, Australia, and is currently with the University of New South Wales, Canberra, Australia. Email: \{zohair.abushaban@unsw.edu.au\}
Xiangyun Zhou, and Thushara Abhayapala are with RSEng, ANU. Emails: \{xiangyun.zhou, thushara.abhayapala\}@anu.edu.au. Henk Wymeersch is with the Department of Signals and Systems, Chalmers University of Technology, Sweden. Email: henkw@chalmers.se. Gonzalo Seco-Granados is with the Department of Telecommunications and Systems Engineering, Universitat Aut\`onoma de Barcelona, Spain (UAB). Email: gonzalo.seco@uab.es.

This work is partly supported by the Horizon2020 project HIGHTS (High precision positioning for cooperative ITS applications) MG-3.5a-2014-636537, the VINNOVA COPPLAR project, funded under Strategic Vehicle Research and Innovation Grant No. 2015-04849, the Australian Government's Research Training Program (RTP) and the Spanish Grant TEC2014-53656-R
}}
\setlength{\abovedisplayskip}{3pt}
\setlength{\belowdisplayskip}{3pt}
\begin{abstract}
Location-aware communication systems are expected to play a pivotal part in the next generation of mobile communication networks. Therefore, there is a need to understand the localization limits in these networks, particularly, using millimeter-wave technology (mmWave). Towards that, we address the uplink and downlink localization limits in terms of 3D position and orientation error bounds for mmWave multipath channels. We also carry out a detailed analysis of the dependence of the bounds on different system parameters. Our key findings indicate that the uplink and downlink behave differently in two distinct ways. First of all, the error bounds have different scaling factors with respect to the number of antennas in the uplink and downlink. Secondly, uplink localization is sensitive to the orientation angle of the user equipment (UE), whereas downlink is not. Moreover, in the considered outdoor scenarios, the non-line-of-sight paths generally improve localization when a line-of-sight path exists. Finally, our numerical results show that mmWave systems are capable of localizing a UE with sub-meter position error, and sub-degree orientation error.
\end{abstract}
\vspace{-7mm}
\section{Introduction}
A strong candidate for the fifth generation of mobile communication (5G) is millimeter-wave (mmWave)  multiple-input-multiple-output (MIMO) technology, where both the user equipment (UE) and the base station (BS) are equipped with arrays of large number of antennas, and operate at a carrier frequency in the range of 30--300 GHz \cite{Andrews2014,Pi2011,Rappaport2013,Heath2016,Orhan2015}. Location-aided systems are expected to have a wide range of applications in 5G mmWave communication \cite{Taranto2014}, whether for vehicular communications \cite{Garcia2016}, assisted living applications \cite{Witrisal2016}, or to support the communication robustness and effectiveness in different aspects such as resource allocation \cite{Muppirisetty2016}, beamforming \cite{Han2012,Aviles2016}, and pilot assignment \cite{Akbar2016}. Therefore, the study of positioning in 5G mmWave systems becomes specially imperative. Due to the use of directional beamforming in mmWave, in addition to the UE position also the UE \textit{orientation} plays an important role in location-aided systems.

Conventionally position information is obtained by GPS, though this has several limitations. Most importantly, GPS suffers from degraded performance in outdoor rich-scattering scenarios and urban canyons, and may fail to provide a position fix for  indoor scenarios. Even in good conditions, GPS positioning accuracy ranges between 1--5 meters. To address these limitations, there has been intense research on competing radio-based localization technologies. To understand the fundamental behavior of any technology, the Cram\'er-Rao lower bound (CRLB)\cite{kay1993} or related bounds can be used. The CRLB provides a lower bound on the variance of an unbiased estimator of a certain parameter. The square-root of the CRLB of the position and the orientation are termed the position error bound (PEB), and the orientation error bound (OEB), respectively. PEB and OEB can be computed indirectly by transforming the bounds of the channel parameters, namely: directions of arrival (DOA), directions of departure (DOD), and time of arrival (TOA).


For conventional MIMO systems, the bounds of the 2D channel parameters are derived in \cite{Li2005}, based on received digital signals and uniform linear arrays (ULA), while bounds are derived in \cite{Larsen2009} based on 3D channel matrix with no transmit beamforming. It was found that having more transmit and receive antennas is beneficial for estimating the DOA and DOD. In both \cite{Li2005,Larsen2009} beamforming was not considered. The bounds on the channel parameters can be transformed into PEB and OEB as in \cite{Shen2007,Shen2010,Shen2010_2,Shen2010_3} that considered 2D cooperative wideband localization, highlighting the benefit of large bandwidths.

MmWave communication combines large antenna arrays with large bandwidths, and should be promising for localization, then. In \cite{Shahmansoori2017}, PEB and OEB  for 2D mmWave downlink localization using ULA are reported, while \cite{Han2016} considers 2D uplink multi-anchor localization. Furthermore, for indoor scenarios, PEB and OEB are analyzed in \cite{Guerra2017} for 3D mmWave uplink localization with a single beam whose direction is assumed to be known. Multipath environments are considered in the works in \cite{Shahmansoori2017,Han2016,Guerra2017}. However, the difference between the uplink and downlink for 3D and 2D with large number of antennas and analog transmit beamforming is yet to be investigated.

In this paper, we consider 3D mmWave localization for both the uplink and downlink under multipath conditions, and derive and analyze the PEB and OEB using multi-beam directional beamforming with arbitrary array geometry. \textcolor{black}{By nature, these bounds are theoretical, and serve as benchmarks to assess location estimation techniques, and study the feasibility of how well the location and orientation can be potentially estimated.} We derive these bounds by transforming the Fisher information matrix (FIM) of the channel parameters into the FIM of location parameters. We stress that although the channel parameters FIM is \textcolor{black}{structured similarly in the uplink and downlink, this is untrue for the location parameters FIM, which is obtained by transforming the FIM of the channel parameters. This transformation is different in the uplink and downlink} and leads  to different PEB and OEB. The contributions of this paper are summarized as follows:
\begin{itemize}
	\item \textcolor{black}{Based on the low-scattering sparse nature of the mmWave channel and the resulting geometrical model, we show that, under some conditions, the multipath parameter estimation can be reduced to a problem of multiple single-path estimation. We refer to  this reduction as the approximate approach. These conditions are highly relevant in mmWave due to channel sparsity, high number of antennas at the receiver and transmitter, and the very large bandwidth.} 
	\item We derive the single-path CRLB of the channel parameters in a closed-form for arbitrary geometry, and show how these bounds are related to the PEB and OEB bounds. We also propose closed-form expressions of PEB and OEB for 3D and 2D line-of-sight (LOS) localization.
	\item We derive the PEB and OEB for general uplink and downlink localization, based on exact and approximate approaches, and show the asymmetry between uplink and downlink via both analytical scaling results and numerical simulations with a uniform rectangular array (URA).
\end{itemize}
The rest of this paper is organized as follows. Section \ref{sec:prob_form} presents the problem statement, before deriving the channel parameters FIM in Section \ref{sec:FIM}, for arbitrary arrays and 3D localization. The transformation of the channel parameters FIM into PEB and OEB is detailed in Section, \ref{sec:PEB_OEB}. The numerical results and the conclusions are reported in Sections \ref{sec:sim} and \ref{sec:conc}, respectively.
\section{Problem Formulation}\label{sec:prob_form}
\begin{figure}[!t]
	\centering
			\vspace{-9mm}
		\includegraphics{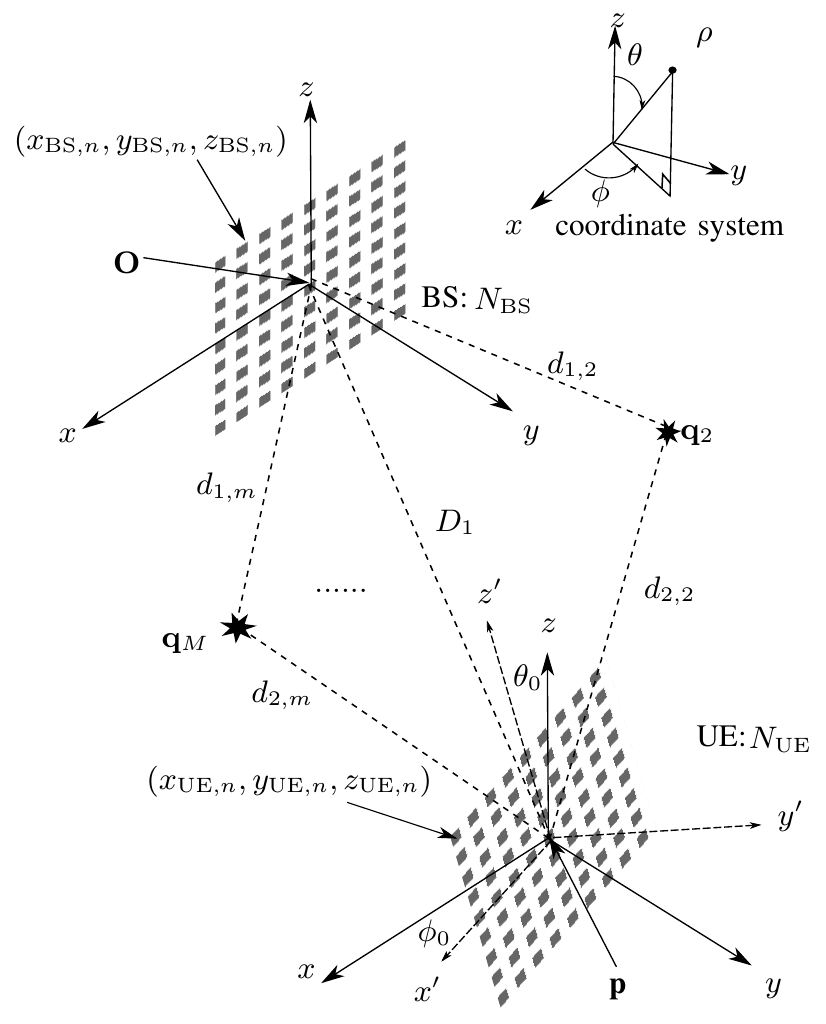}
				\vspace{-6mm}
	\caption{An example scenario composed of a URA of $N_\mathrm{UE}=N_\mathrm{BS}=81$ antennas, and $M$ paths. We use the spherical coordinate system highlighted in the top right corner. The axes labeled $x',y',z'$ are the result of rotating the array by $(\theta_0, \phi_0)$.}\label{fig:scenario}
	\vspace{-10mm}
\end{figure}

\subsection{System Geometry}
We consider a BS equipped with an array of $N_\mathrm{BS}$ antennas arranged in an arbitrary but known geometry whose centroid, i.e., \textcolor{black}{geometric center},  is located at the origin ($\mathbf{O}$), and orientation angle is $\mathbf{o}_{\mathrm{BS}}=[0, 0]^\mathrm{T}$. On the other hand, the centroid of the UE is located at an unknown position $\mathbf{p}=[p_x,p_y,p_z]^\mathrm{T}$ and equipped with a second array of $N_\mathrm{UE}$ antennas arranged in an arbitrary but known geometry with an unknown orientation $\mathbf{o}=[\theta_0, \phi_0]^\mathrm{T}$, aligning the UE with the rotated axes $x',y'$ and $z'$. An example with URAs is shown in Fig. \ref{fig:scenario}. \textcolor{black}{ $\phi_0$ is defined as the rotation around the $z$-axis, while $\theta_0$ is defined as the rotation around the $-x'$-axis. Thus, the UE array elements locations is obtained using a rotation matrix $\mathbf{R}(\theta_0,\phi_0)$ derived in Section \ref{sec:rotation_matrix}. Considering two orientation angles is highly relevant in applications such as vehicular communication and robotics, where the UE turns left and right, or goes up or down hills, without rotating the vehicle axis.}   There are $M \ge 1$ paths between BS and UE, where the first path is LOS, while with the other $M-1$ paths are associated with clusters located at $\mathbf{q}_m=[q_{m,x},q_{m,y},q_{m,z}]^\mathrm{T}$, $2\leq{m}\leq{M}$. These clusters can be reflectors or scatterers. Due to the mmWave propagation characteristics, the number of paths is small \cite{Pi2011} and correspond to single-bounce reflections \cite{Deng2014,Shahmansoori2017}. It was experimentally observed in \cite{Akdeniz2014} that the average value of $M$ under mmWave propagation in an urban environment in New York City is 2 to 3 paths with a maximum of 4 paths present. This resulted in modeling $M$ by a Poisson mass function whose average is 1.8  and 1.9 at 28 GHz, and 73 GHz, respectively. Thus, the channel is considered spatially sparse, and the parameters of different paths are assumed to be distinct, i.e., unique DOAs, DODs, TOAs and channel gains.
\subsection{Channel Model}
Denote the $m^{\mathrm{th}}$ DOD and DOA by $(\theta_{\mathrm{T},m},\phi_{\mathrm{T},m}) \text{ and }(\theta_{\mathrm{R},m},\phi_{\mathrm{R},m}), 1\leq{m}\leq{M}$, respectively, where the related unit-norm array response vectors are given by \cite{VanTrees2002}
\begin{align}
\mathbf{a}_{\mathrm{T},m}(\theta_{\mathrm{T},m},\phi_{\mathrm{T},m})&\triangleq\frac{1}{\sqrt{N_{\mathrm{T}}}}e^{-j\boldsymbol{\Delta}_{\mathrm{T}}^\mathrm{T}\mathbf{k}(\theta_{\mathrm{T},m},\phi_{\mathrm{T},m})},\qquad\in\mathbb{C}^{N_{\mathrm{T}}}\label{eq:a_t}\\
\mathbf{a}_{\mathrm{R},m}(\theta_{\mathrm{R},m},\phi_{\mathrm{R},m})&\triangleq\frac{1}{\sqrt{N_{\mathrm{R}}}}e^{-j\boldsymbol{\Delta}_{\mathrm{R}}^\mathrm{T}\mathbf{k}(\theta_{\mathrm{R},m},\phi_{\mathrm{R},m})},\qquad\in\mathbb{C}^{N_{\mathrm{R}}}\label{eq:a_r}
\end{align}
where  $\mathbf{k}(\theta,\phi)=\frac{2\pi}{\lambda}[\cos\phi\sin\theta, \sin\phi\sin\theta, \cos\theta]^\mathrm{T}$ is the wavenumber vector, $\lambda$ is the wavelength, $\boldsymbol\Delta_{\mathrm{R}}\triangleq[\mathbf{u}_{\mathrm{R},1},\mathbf{u}_{\mathrm{R},2},\cdots,\mathbf{u}_{\mathrm{R},N_{\mathrm{R}}}]$, $\mathbf{u}_{\mathrm{R},n}\triangleq[x_{\mathrm{R},n},y_{\mathrm{R},n},z_{\mathrm{R},n}]^\mathrm{T}$ is a vector of Cartesian coordinates of the $n^{\mathrm{th}}$ receiver element,  and $N_{\mathrm{R}}$ is the number of receiving antennas. $N_{\mathrm{T}}$, $\boldsymbol\Delta_{\mathrm{T}}$ and $\mathbf{u}_{\mathrm{T},n}$ are defined similarly\footnote{The subscripts $\mathrm{T}$ and $\mathrm{R}$ refer to the transmit and receive sides, respectively, regardless of using the uplink or downlink. On the other hand, when the notation is unique to the base station or the user equipment, we use the subscript $\mathrm{BS}$ and $\mathrm{UE}$.}. We drop the angle parameters from the notation of $\mathbf{a}_{\mathrm{T},m}, \text{ and }\mathbf{a}_{\mathrm{R},m}$.

Denoting the TOA of the $m^{\mathrm{th}}$ path by $\tau_m$, the channel can be expressed\footnote{We use a narrow-band array model, so that $A_\text{max}\ll c/W$, where $A_\text{max}$ is maximum array aperture, $c$ is speed of light, and $W$ is the system bandwidth \cite{VanTrees2002}.} as
\begin{align}
\mathbf{H}(t)& = \sum_{m=1}^{M}\mathbf{H}_m\delta (t-{\tau_m}),\label{eq:channelmodel}
\end{align}
From Fig. \ref{fig:scenario}, \textcolor{black}{assuming synchronization}\footnote{\textcolor{black}{We rely on the commonly used synchronization assumption e.g., \cite{Witrisal2016}, \cite{Li2005}, \cite{Shen2007,Shen2010,Shen2010_2,Shen2010_3}, \cite{Han2016}, to gain fundamental understanding. We, however, realize that in practice synchronization errors must be accounted for in protocols and algorithms. This can be done by, e.g., a two-way protocol or a joint localization and synchronization approach, which is beyond the scope of this paper.}}, $\tau_m=D_m/c$, where $D_m=d_{1,m}+d_{2,m}$, for $m>1$ and
\begin{align}
\mathbf{H}_m\triangleq\sqrt{N_{\mathrm{R}}N_\mathrm{T}}\beta_m\mathbf{a}_{\mathrm{R},m}\mathbf{a}^\mathrm{H}_{\mathrm{T},m}\in\mathbb{C}^{N_{\mathrm{R}}\times{N_\mathrm{T}}},
\end{align}
where $\beta_m$ is the complex gain of the $m^{\mathrm{th}}$ path. Finally, we define the following vectors
 \begin{equation*}
 \begin{aligned}
 \boldsymbol{\theta}_{\mathrm{R}}&\triangleq[\theta_{\mathrm{R},1},\theta_{\mathrm{R},2},...,\theta_{\mathrm{R},M}]^\mathrm{T},\\
  \boldsymbol{\phi}_{\mathrm{R}}&\triangleq[\phi_{\mathrm{R},1},\phi_{\mathrm{R},2},...,\phi_{\mathrm{R},M}]^\mathrm{T},\\
 \end{aligned}\qquad
 \begin{aligned}
  \boldsymbol{\theta}_\mathrm{T}&\triangleq[\theta_{\mathrm{T},1},\theta_{\mathrm{T},2},...,\theta_{\mathrm{T},M}]^\mathrm{T},\\ \boldsymbol{\phi}_\mathrm{T}&\triangleq[\phi_{\mathrm{T},1},\phi_{\mathrm{T},2},...,\phi_{\mathrm{T},M}]^\mathrm{T},\\
    \end{aligned}\qquad
   \begin{aligned}
   \boldsymbol{\beta}&\triangleq[\beta_1, \beta_2, ..., \beta_M]^\mathrm{T},\\   \boldsymbol{\tau}&\triangleq[\tau_{1},\tau_{2},...,\tau_{M}]^\mathrm{T}.
   \end{aligned}
\end{equation*}

\subsection{Transmission Model}
The transmitter sends a signal $\sqrt{E_\mathrm{s}} \mathbf{Fs}(t)$, where $E_\mathrm{s}$ is the transmitted energy per symbol duration, $\mathbf{F}\triangleq[\mathbf{f}_1,\mathbf{f}_2,...\mathbf{f}_{N_{\mathrm{B}}}]$ is a directional beamforming matrix, where \begin{align}
\mathbf{f}_\ell=\frac{1}{\sqrt{N_{\mathrm{B}}}}\mathbf{a}_{\mathrm{T},\ell}(\theta_\ell,\phi_\ell), \ \  1\leq\ell\leq{N_{\mathrm{B}}}\label{eq:BF_Definition}
\end{align}
is a beam pointing towards $(\theta_\ell,\phi_\ell)$ of the same form as \eqref{eq:a_t}, and $N_{\mathrm{B}}$ is the number of transmitted beams. The pilot signal $\mathbf{s}(t)\triangleq[s_1(t),s_2(t),...,s_{N_{\mathrm{B}}}(t)]^\mathrm{T}$ is expressed as
\begin{align}
s_\ell(t)=\sum_{k=0}^{N_{\mathrm{s}}-1}a_{\ell,k}p(t-kT_{\mathrm{s}}),\ 1\leq\ell\leq{N_{\mathrm{B}}}, \label{eq:txsignal}
\end{align}
where $a_{\ell,k}$ are known unit-energy pilot symbols transmitted over the $\ell^{\mathrm{th}}$ beam, and  $p(t)$, is a unit-energy pulse with a power spectral density (PSD), denoted by $|P(f)|^2$. In \eqref{eq:txsignal}, $N_{\mathrm{s}}$ is the number of pilot symbols and $T_{\mathrm{s}}$ is the symbol duration, leading to a total observation time of  $T_{\mathrm{o}} \approx N_{\mathrm{s}}T_{\mathrm{s}}$. To keep the transmitted power fixed with $N_\mathrm{T}$ and $N_\mathrm{B}$, we set $\mathrm{Tr}\left(\mathbf{F}^\mathrm{H}\mathbf{F}\right)=1$,  $\mathbb{E}\{\mathbf{s}(t)\mathbf{s}^\mathrm{H}(t)\}=\mathbf{I}_{N_{\mathrm{B}}}$, where $\mathrm{Tr}\left(\cdot\right)$ denotes the matrix trace, and $\mathbf{I}_{N_{\mathrm{B}}}$ is the $N_{\mathrm{B}}$-dimensional identity matrix.

The  received signal observed at the  input of the receive beamformer is given by
\begin{align}
\mathbf{r}(t)&\triangleq\sum_{m=1}^{M}\sqrt{E_\mathrm{s}}\mathbf{H}_m\mathbf{Fs}(t-{\tau_m})+\mathbf{n}(t),\ \  t\in[0,T_{\mathrm{o}}],\label{eq:model}
\end{align}
where $\mathbf{n}(t)\triangleq[n_1(t),n_2(t),...,n_{N_\mathrm{R}}(t)]^\mathrm{T}\in\mathbb{C}^{N_{\mathrm{R}}}$ is zero-mean white Gaussian noise with PSD $N_0$. Similar to \cite{Fakharzadeh2010,Venkateswaran2010}, we assume that a low-noise amplifier and a passband filter are attached to each receive antenna. While this may seem a restrictive assumption, it will allow us to derive the PEB and OEB, which are fundamental lower bounds irrespective of the type of processing performed at the receiver, such as receive beamforming.
\subsection{3D Localization Problem}
Our objective is to derive the UE PEB and OEB, based on the observed signal, $\mathbf{r}(t)$, for both the uplink and downlink. We achieve this in two steps: firstly, we derive bounds on the channel parameters, namely, direction of arrival, $(\boldsymbol{\theta}_{\mathrm{R}},\boldsymbol{\phi}_{\mathrm{R}}),$ direction of departure, $(\boldsymbol{\theta}_\mathrm{T},\boldsymbol{\phi}_\mathrm{T}),$ time of arrival $\boldsymbol{\tau}$, and paths gains, $\boldsymbol{\beta}$. Secondly, we transform these bounds into the position domain. 

\section{Fisher Information Matrix of The Channel Parameters}\label{sec:FIM}
We first derive exact expressions for the entries of the FIM. Then, we determine the conditions under which the individual paths can be considered orthogonal. Subsequently, we provide closed-form expressions of the CRLB for the single-path case for 3D and 2D localization.
%
%

\subsection{Exact Expression}\label{sec:exact_crlb}
Let us define the parameter vector
\begin{align}
\boldsymbol{\varphi}\triangleq[\boldsymbol{\theta}_{\mathrm{R}}^\mathrm{T},\boldsymbol{\theta}_\mathrm{T}^\mathrm{T},\boldsymbol{\phi}_{\mathrm{R}}^\mathrm{T},\boldsymbol{\phi}_\mathrm{T}^\mathrm{T},\boldsymbol{\tau}^\mathrm{T},\boldsymbol{\beta}_{\mathrm{R}}^\mathrm{T},\boldsymbol{\beta}_{\mathrm{I}}^\mathrm{T}]^\mathrm{T},
\end{align}
where
$\boldsymbol{\beta}_{\mathrm{R}}\triangleq\Re\{\boldsymbol{\beta}\}$, and  $\boldsymbol{\beta}_{\mathrm{I}}\triangleq\Im\{\boldsymbol{\beta}\}$ are the real and imaginary parts of $\boldsymbol{\beta}$, respectively, and $\varphi_u$ is the $u^{\mathrm{th}}$ element in $\boldsymbol\varphi$. Then, the corresponding FIM, partitioned into $M\times{M}$ submatrices, is

 structured as
\begin{align}
\label{eq:matrix_I}
\mathbf{J}_{\boldsymbol\varphi}\triangleq\begin{bmatrix}
\mathbf{J}_{{\boldsymbol\theta}_{\mathrm{R}}{\boldsymbol\theta}_{\mathrm{R}}}&\mathbf{J}_{{\boldsymbol\theta}_{\mathrm{R}}{\boldsymbol\theta}_{\mathrm{T}}}&\cdots&\mathbf{J}_{{\boldsymbol\theta}_{\mathrm{R}}{{\boldsymbol\beta}_{\mathrm{I}}}}\\
\mathbf{J}_{{\boldsymbol\theta}_{\mathrm{R}}{\boldsymbol\theta}_{\mathrm{T}}}^\mathrm{T}&\ddots&\cdots&\vdots\\
\vdots&\cdots&\ddots&\vdots\\
\mathbf{J}_{{{\boldsymbol\theta}_{\mathrm{R}}}\boldsymbol\tau}^\mathrm{T}&\cdots&\cdots&\mathbf{J}_{{{\boldsymbol\beta}_{\mathrm{I}}}{{\boldsymbol\beta}_{\mathrm{I}}}}\\
\end{bmatrix},
\end{align}
where, due to the additive white Gaussian noise \cite{kay1993},
\begin{align}
[\mathbf{J}_{\boldsymbol\varphi}]_{u,v}&\triangleq\frac{1}{N_0}\int_{0}^{T_{\mathrm{o}}}\Re\left\lbrace\frac{\partial\boldsymbol{\mu}^\mathrm{H}_{\boldsymbol{\varphi}}(t)}{\partial{\varphi_u}}\frac{\partial\boldsymbol{\mu}_{\boldsymbol\varphi}(t)}{\partial{\varphi_v}} \right\rbrace \mathrm{d}t,\label{eq:FIM}
\end{align}
\textcolor{black}{where $\boldsymbol{\mu}_{\boldsymbol\varphi}(t)$ is the noiseless part of received signal in \eqref{eq:model}}, given by
\begin{align}
\boldsymbol{\mu}_{\boldsymbol\varphi}(t)&\triangleq\sqrt{N_{\mathrm{R}}N_\mathrm{T}E_\mathrm{s}}\sum_{m=1}^{M}\beta_m\mathbf{a}_{\mathrm{R},m}\mathbf{a}^\mathrm{H}_{\mathrm{T},m}\mathbf{Fs}(t-{\tau_m}).\label{eq:multipath_mu}
\end{align}
We now introduce the following matrices to simplify the notation
 \label{eq:RXmatrices}
	\begin{align}
	\begin{aligned}
	\tilde{\mathbf{K}}_{\mathrm{R},m}&\triangleq\diag\left(\frac{\partial}{\partial{\theta_{\mathrm{R},m}}}\boldsymbol\Delta_{\mathrm{R}}^\mathrm{T}\mathbf{k}(\theta_{\mathrm{R},m},\phi_{\mathrm{R},m})\right),\\
		{\mathbf{P}}_{\mathrm{R}}&\triangleq[\tilde{\mathbf{P}}_{\mathrm{R},1}\mathbf{a}_{\mathrm{R},1}, \tilde{\mathbf{P}}_{\mathrm{R},2}\mathbf{a}_{\mathrm{R},2}, ..., \tilde{\mathbf{P}}_{\mathrm{R},N_{\mathrm{R}}}\mathbf{a}_{\mathrm{R},N_{\mathrm{R}}}],\\
			\mathbf{A}_{\mathrm{R}}&\triangleq[\mathbf{a}_{\mathrm{R},1}, \mathbf{a}_{\mathrm{R},2}, ..., \mathbf{a}_{\mathrm{R},N_{\mathrm{R}}}],
	\end{aligned}
	\quad
	\begin{aligned}
	\tilde{\mathbf{P}}_{\mathrm{R},m}&\triangleq\diag\left(\frac{\partial}{\partial{\phi_{\mathrm{R},m}}}\boldsymbol\Delta_{\mathrm{R}}^\mathrm{T}\mathbf{k}(\theta_{\mathrm{R},m},\phi_{\mathrm{R},m})\right),\\
	{\mathbf{K}}_{\mathrm{R}}&\triangleq[\tilde{\mathbf{K}}_{\mathrm{R},1}\mathbf{a}_{\mathrm{R},1}, \tilde{\mathbf{K}}_{\mathrm{R},2}\mathbf{a}_{\mathrm{R},2}, ..., \tilde{\mathbf{K}}_{\mathrm{R},N_{\mathrm{R}}}\mathbf{a}_{\mathrm{R},N_{\mathrm{R}}}],\\
	\mathbf{B}&\triangleq{\diag}(\boldsymbol{\beta})
	\end{aligned}
	\end{align}
with similar expression, obtained by replacing ``R" with ``T". \textcolor{black}{It is shown in Appendix \ref{app:FIM_entries}} that each submatrix in (\ref{eq:matrix_I}) is of the form
\begin{align}
\mathbf{J}_{\mathbf{x},\mathbf{x'}}= \Re \left\{ (\mathrm{RX~factor})\odot(\mathrm{TX~factor})\odot(\mathrm{signal~factor})\right\}, \label{eq:refcase}
\end{align}
where  $\odot$ denotes the Hadamard product,  the RX factor relates to the receiver array, the TX factor relates to the transmitter array and beamforming, and \textcolor{black} {the signal} factor relates to the pilot signals. The RX factor is a product of the matrices $\{{\mathbf{A}}_{\mathrm{R}}\mathbf{B},{\mathbf{K}}_{\mathrm{R}}\mathbf{B},{\mathbf{P}}_{\mathrm{R}}\mathbf{B}\}$, while the TX part is a product of similar matrices $\{\mathbf{F}^\mathrm{H}\mathbf{A}_\mathrm{T},\mathbf{F}^\mathrm{H}{\mathbf{K}}_{\mathrm{T}},\mathbf{F}^\mathrm{H}{\mathbf{P}}_{\mathrm{T}}\}$, associated with the transmitter as well as $\mathbf{F}$. Under the assumption of i.i.d. symbols, the signal factor depends on
\begin{align}
[\mathbf{R}_i]_{uv}\triangleq\int_{-W/2}^{W/2}(2\pi f)^{i}|P(f)|^2e^{-j2\pi f\Delta\tau_{uv}}\mathrm{d}f, \label{eq:R_i}
\end{align}
in which $\Delta\tau_{uv}\triangleq\tau_v-\tau_u$, $i\in \{0,1,2\}$, and $W$ is the signal bandwidth. \textcolor{black}{The signal factor in \eqref{eq:R_i} represents the correlation between different paths obtained in the frequency domain using Parseval's theorem. 
	See \eqref{eq:R_0}, \eqref{eq:R_1} and \eqref{eq:R_2}}. For instance, \textcolor{black}{defining the signal-to-noise ratio (SNR) as} $\gamma\triangleq{N_{\mathrm{R}}N_\mathrm{T}N_\mathrm{s}E_\mathrm{s}/N_0}$, it can be verified that
\begin{align}
& \mathbf{J}_{{\boldsymbol\theta}_{\mathrm{R}}{\boldsymbol\theta}_{\mathrm{R}}}=\gamma\Re \left\{ (\mathbf{B}^\mathrm{H}{\mathbf{K}}_{\mathrm{R}}^{\mathrm{H}}{\mathbf{K}}_{\mathrm{R}}\mathbf{B})\odot(\mathbf{A}_\mathrm{T}^{\mathrm{H}}\mathbf{F}\mathbf{F}^\mathrm{H}\mathbf{A}_\mathrm{T})^\mathrm{T}\odot\mathbf{R}_0\right\}\label{eq:thRthR}.
\end{align}
The remaining entries of (\ref{eq:matrix_I}) exhibit the structure in \eqref{eq:refcase}, as listed in Appendix \ref{app:FIM_entries}. \textcolor{black}{Note that the FIMs in (\ref{eq:matrix_I}) and (\ref{eq:bigFIMS}) scale linearly with SNR, meaning that CRLB decreases as SNR increases.}

\subsection{Approximate Fisher Information of the Channel Parameters}\label{sec:approx_crlb}
The exact FIMs presented in Section \ref{sec:exact_crlb} provide the exact CRLB of the channel parameters. However, under some circumstances, it is possible to simplify this computation by reducing the submatrices of the FIMs to either diagonal or zero matrices, by exploiting the structure in \eqref{eq:refcase}. Inspired by \cite{Wilkinson1968}, we start by introducing the following definition
\begin{definition}
\textcolor{black}{Given a square matrix $\mathbf{A}(\kappa)$ that can be decomposed into a diagonal matrix $\mathbf{D}(\kappa)\neq\mathbf{0}$ plus a hollow matrix $\mathbf{E}(\kappa)\neq\mathbf{0}$, then $\mathbf{A}(\kappa)=\mathbf{D}(\kappa)+\mathbf{E}(\kappa)$} is {almost diagonal} (AD) for any  parameter $\kappa$ if
\begin{align}
\lim_{\kappa\rightarrow\infty}\delta(\mathbf{A},\kappa)\triangleq\lim_{\kappa\rightarrow\infty}\frac{\|\mathbf{E}(\kappa)\|_\mathrm{F}}{\|\mathbf{D}(\kappa)\|_\mathrm{F}}=0,
\end{align}
where $\|\cdot\|_\mathrm{F}$ denotes the Frobenius norm.
\end{definition}

We now use Definition 1 to inspect the factors in \eqref{eq:refcase}, and understand the behavior under typical mmWave conditions, i.e., large transmit and receive arrays and large system bandwidth.

\begin{itemize}
\item \textit{Factor 1 -- Receiver Side: }
For a large number of receive antennas, the power received from a direction $(\theta_{\mathrm{R},u}, \phi_{\mathrm{R},u})$ via a steering vector in the direction $(\theta_{\mathrm{R},v}, \phi_{\mathrm{R},v})$ is very small, i.e., $\|\mathbf{a}_{\mathrm{R},u}^\mathrm{H}\mathbf{a}_{\mathrm{R},v}\|\ll\|\mathbf{a}_{\mathrm{R},u}\|^2$, $u\neq{v}$.  Thus, $\lim_{N_{\mathrm{R}}\rightarrow\infty}\delta(\mathbf{A}_{\mathrm{R}}^{\mathrm{H}}\mathbf{A}_{\mathrm{R}},N_{\mathrm{R}})={0}$. Similarly, considering the exponential form of $\mathbf{a}_{\mathrm{R},m}$ and that  $\tilde{\mathbf{K}}_{\mathrm{R},m}, \tilde{\mathbf{P}}_{\mathrm{R},m}$ are diagonal,
\begin{subequations}
\begin{align}
\lim_{N_{\mathrm{R}}\rightarrow\infty}\delta({\mathbf{K}}^{\mathrm{H}}_R{\mathbf{K}}_\mathrm{R},N_{\mathrm{R}})&=\lim_{N_{\mathrm{R}}\rightarrow\infty}\delta({\mathbf{P}}^{\mathrm{H}}_R{\mathbf{P}}_\mathrm{R},N_{\mathrm{R}})=\lim_{N_{\mathrm{R}}\rightarrow\infty}\delta({\mathbf{K}}^{\mathrm{H}}_\mathrm{R}{\mathbf{P}}_\mathrm{R},N_{\mathrm{R}})={0}.\notag
\end{align}
\end{subequations}
On the other hand, using the facts that the BS centroid is at the origin, then for uplink
\textcolor{black}{$\Tr(\tilde{\mathbf{K}}_{\mathrm{R},m})=\sum_{n=0}^{N_\mathrm{R}}\frac{\partial}{\partial\theta_\mathrm{R}}\boldsymbol{\Delta}_\mathrm{R}^\mathrm{T}\mathbf{k}=\left(\sum_{n=0}^{N_\mathrm{R}}\boldsymbol{\Delta}_\mathrm{R}^\mathrm{T}\right)\frac{\partial}{\partial\theta_\mathrm{R}}\mathbf{k}=0,$} and similarly, $\Tr(\tilde{\mathbf{P}}_{\mathrm{R},m})=0$, and that the UE centroid is at $\mathbf{p}$, then for downlink $ \Tr(\tilde{\mathbf{K}}_{\mathrm{R},m})=\mathbf{p}^\mathrm{T} \frac{\partial}{\partial\theta_m}{\mathbf{k}}$,$~
 \Tr(\tilde{\mathbf{P}}_{\mathrm{R},m})=\mathbf{p}^\mathrm{T}\frac{\partial}{\partial\phi_m}\mathbf{k}.$
Moreover, Noting that  $[{\mathbf{K}}^\mathrm{H}_{\mathrm{R}}{\mathbf{A}}_{\mathrm{R}}]_{m,m}={\Tr(\tilde{\mathbf{K}}_{\mathrm{R},m})}/{N_{\mathrm{R}}}$, then, for both uplink or downlink,
\begin{align}
\lim_{N_{\mathrm{R}}\rightarrow\infty}{\mathbf{K}}^\mathrm{H}_{\mathrm{R}}{\mathbf{A}}_{\mathrm{R}}&=\lim_{N_{\mathrm{R}}\rightarrow\infty}{\mathbf{P}}^\mathrm{H}_{\mathrm{R}}{\mathbf{A}}_{\mathrm{R}}=\mathbf{0}_{M\times{M}}.
\end{align}
where $\mathbf{0}_{M\times{M}}$ is an $M\times{M}$ matrix of zeros.

\item \textit{\textcolor{black}{Factor 2 -- Transmitter Side:}} The transmitter side contributes to the FIM in (\ref{eq:thRthR}) by $\mathbf{A}_\mathrm{T}^{\mathrm{H}}\mathbf{FF}^\mathrm{H}\mathbf{A}_\mathrm{T}$. Recalling that $[\mathbf{A}_\mathrm{T}^{\mathrm{H}}\mathbf{FF}^\mathrm{H}\mathbf{A}_\mathrm{T}]_{u,v}=\mathbf{a}_{\mathrm{T},u}^\mathrm{H}\mathbf{FF}^\mathrm{H}\mathbf{a}_{\mathrm{T},v},$
the right-hand side term can be interpreted as the spatial cross-correlation between the $u^{\mathrm{th}}$ and the $v^{\mathrm{th}}$ DODs. So, considering directional beamforming, as $N_\mathrm{T}$ increases, the beams become narrower, leading to
\begin{align}\label{eq:factor2}
\lim_{N_\mathrm{T}\rightarrow\infty}\mathbf{a}_{\mathrm{T},u}^\mathrm{H}\mathbf{FF}^\mathrm{H}\mathbf{a}_{\mathrm{T},v}\approx 0\qquad u\neq v.
\end{align}
Moreover, for extremely narrow beams, the likelihood of covering the $v^{\mathrm{th}}$ DOD is almost zero,
\begin{align}\label{eq:factor2_2}
\lim_{N_\mathrm{T}\rightarrow\infty}\mathbf{a}_{\mathrm{T},v}^\mathrm{H}\mathbf{FF}^\mathrm{H}\mathbf{a}_{\mathrm{T},v}=\lim_{N_\mathrm{T}\rightarrow\infty}\|\mathbf{F}^\mathrm{H}\mathbf{a}_{\mathrm{T},v}\|^2\approx 0.
\end{align}
In this extreme case, $\mathbf{A}_\mathrm{T}^{\mathrm{H}}\mathbf{FF}^\mathrm{H}\mathbf{A}_\mathrm{T}\approx\mathbf{0}$, which implies the whole FIM is zero. However, since \eqref{eq:factor2} approaches 0 faster than \eqref{eq:factor2_2}, the transmission over directional beamforming should be restricted to ${N_\mathrm{T}<\infty}$. That said, there are cases where \eqref{eq:factor2_2} does not hold and $\mathbf{A}_\mathrm{T}^{\mathrm{H}}\mathbf{FF}^\mathrm{H}\mathbf{A}_\mathrm{T}$ is AD, e.g., when using random beamforming \cite{Zohair2016}.
By inspection, a similar statement can be made for $({\mathbf{P}}_\mathrm{T}^{\mathrm{H}}\mathbf{FF}^\mathrm{H}{\mathbf{P}}_\mathrm{T})^\mathrm{T},({\mathbf{K}}_\mathrm{T}^{\mathrm{H}}\mathbf{FF}^\mathrm{H}{\mathbf{K}}_\mathrm{T})^\mathrm{T},$ and $({\mathbf{P}}_\mathrm{T}^{\mathrm{H}}\mathbf{FF}^\mathrm{H}{\mathbf{K}}_\mathrm{T})^\mathrm{T}$.

\item \textit{Factor 3: Multipath Cross-Correlation}: It can be shown that the cross-correlation functions in (\ref{eq:R_i}) are even \textcolor{black}{in $\Delta\tau_{uv}$} for $i=0,2$ and have maxima on their diagonals. These maxima are constant with values 1 and $4\pi^2W_\mathrm{eff}^2$, respectively, where $W_\mathrm{eff}^2=\int_{-W/2}^{W/2}f^2|P(f)|^2\mathrm{d}f$. Moreover,
\begin{align}
\lim_{W \to \infty} \delta(\mathbf{R}_0,W) = 0,\qquad\qquad\qquad
\lim_{W \to \infty} \delta(\mathbf{R}_2,W) = 0.
\end{align}
For any $W$, $\diag(\mathbf{R}_1)=\mathbf{0}_{M}$ so that $\mathbf{R}_1$ is a hollow matrix, with $\lim_{W \to \infty} \mathbf{R}_1 = \mathbf{0}_{M\times M}$. \textcolor{black}{So, in effect, the paths overlapping in time is negligible, which is consistent with \cite{Shen2008_2} and \cite{Leitinger2015}.}
\end{itemize}
\begin{figure}[!t]
	\centering
			\vspace{-2mm}
	\includegraphics[scale=0.9]{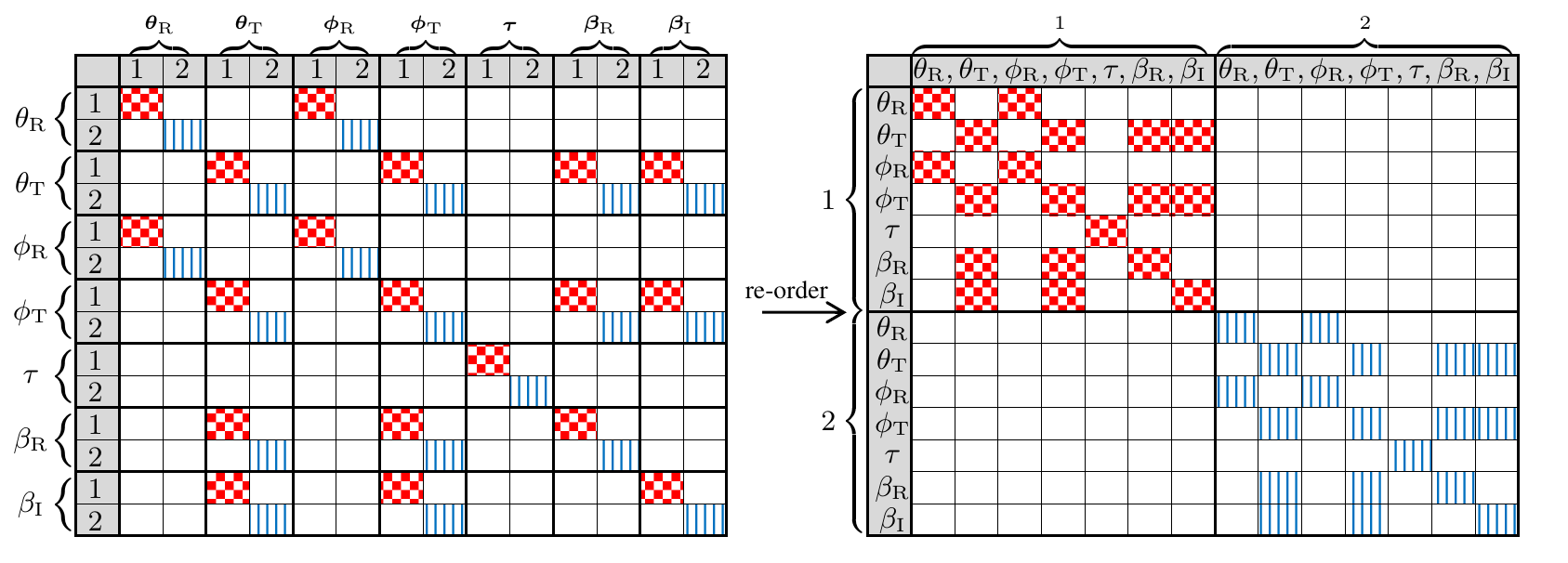}
		\vspace{-4mm}
	\caption{An example on the approximate FIM in (\ref{eq:matrix_I}) with $M=2$. The red-pattern and blue-strip cells represent the non-zero entries of the FIM and correspond to $m=1$, and $m=2$, respectively. Re-ordering the FIM on the left yields the FIMs in (\ref{eq:FIM_reordered}).}\label{fig:FIM_reorder}
	\vspace{-10mm}
\end{figure}

In combination, given the Hadamard product of \eqref{eq:refcase}, and assuming $N_\mathrm{T}<\infty$, we find that under typical mmWave conditions, due to the combined effect of large values of $N_\mathrm{R}$, and $W$, some submatrices of the FIM in ($\ref{eq:matrix_I}$) are AD, while the others are almost zero as shown in Fig.~\ref{fig:FIM_reorder} (left). This effect relates to the fact that paths can be resolved in either direction of arrival, direction, or delay domain, both of which reinforce each other in mmWave.  \textcolor{black}{In other words, it is sufficient to have $N_\mathrm{R}\rightarrow\infty$ or $W\rightarrow\infty$ in order for the paths to be orthogonal. However, as it is customary to have very large antenna arrays in mmWave systems, it is reasonable to assume that paths are always orthogonal.} Re-ordering the parameters, grouping them path by path, we obtain the block diagonal FIM in Fig.~\ref{fig:FIM_reorder} (right).

\subsection{FIM of Single-Path mmWave Channel Parameters}\label{sec:los_crlb}
\textcolor{black}{Focusing on the FIM of the $m^\mathrm{th}$ path, it is interesting to note that after obtaining the FIM in Fig.~\ref{fig:FIM_reorder} (right), it becomes evident that the estimation of $\tau_m$ is independent of any other parameter. This follows from the fact that, for the $m^\mathrm{th}$ path, $\Delta\tau_{mm}=0$ and $[\mathbf{R}_1]_{m,m}=0$. Moreover, note that the estimation of $\theta_{\mathrm{R},m}$ and $\phi_{\mathrm{R},m}$ is independent of the other parameters, unlike $\theta_{\mathrm{T},m}$ and $\phi_{\mathrm{T},m}$ which depend on $\beta_m$. This is because we use transmit beamforming only, hence power gain has two components: channel gain and antenna directional gain.}

We now use the notion of the equivalent FIM (EFIM) from \cite{Shen2010_2} to isolate the effect of the nuisance parameter $\beta$. EFIM is a measure of the information corresponding to a certain unknown parameter, taking into account the uncertainties of the other unknown parameters.
\begin{definition}
	Given a parameter vector $\boldsymbol{\theta}\triangleq[\boldsymbol{\theta}_1^{\mathrm{T}},\boldsymbol{\theta}_2^{\mathrm{T}}]^{\mathrm{T}}$ with corresponding FIM
	\begin{align}
	\mathbf{J}_{\boldsymbol{\theta}\boldsymbol{\theta}}=\begin{bmatrix}\mathbf{J}_{\boldsymbol{\theta}_1\boldsymbol{\theta}_1}&\mathbf{J}_{\boldsymbol{\theta}_1\boldsymbol{\theta}_2}\\
	\mathbf{J}^\mathrm{T}_{\boldsymbol{\theta}_1\boldsymbol{\theta}_2}&\mathbf{J}_{\boldsymbol{\theta}_2\boldsymbol{\theta}_2}
	\end{bmatrix},
	\end{align}
	Then, the EFIM of $\boldsymbol{\theta}_1$ is given by
	\begin{align}
	\mathbf{J}^{\mathrm{e}}_{\boldsymbol{\theta}_1\boldsymbol{\theta}_1}=\mathbf{J}_{\boldsymbol{\theta}_1\boldsymbol{\theta}_1}-\mathbf{J}_{\boldsymbol{\theta}_1\boldsymbol{\theta}_2}\mathbf{J}^{-1}_{\boldsymbol{\theta}_2\boldsymbol{\theta}_2}\mathbf{J}^\mathrm{T}_{\boldsymbol{\theta}_1\boldsymbol{\theta}_2}.
	\end{align}
\end{definition}
Given the block-diagonal  structure of the approximate FIM, it becomes meaningful to study paths separately. Thus, considering $\beta$ as a nuisance parameter, we focus on a single path, with the parameters of interest being ${\boldsymbol\varphi}_\mathrm{ch}\triangleq[\theta_{\mathrm{R}},\theta_\mathrm{T},\phi_{\mathrm{R}},\phi_\mathrm{T},\tau]^\mathrm{T}$, and write the EFIM of the DOA, DOD, and TOA, from Appendix \ref{app:CRLB}, as follows
\begin{align}
\mathbf{J}^\mathrm{e}_{{\boldsymbol\varphi}_\mathrm{ch}}=\begin{bmatrix}\mathbf{J}^\mathrm{e}_{\boldsymbol{\theta,\phi}}&\mathbf{0}_4\\
\mathbf{0}^\mathrm{T}_4&\gamma|\beta|^2 4\pi^2GW^2_{\mathrm{eff}}\end{bmatrix},\label{eq:partitioned}
	\end{align}
	where\small
\begin{align}\label{eq:EFIM_DOD_DOA}
\mathbf{J}^\mathrm{e}_{\boldsymbol{\theta,\phi}}&=\gamma|\beta|^2\left[\begin{array}{cccc}
R_\theta G& 0 &X_{\theta ,\phi}G &0\\
0& \frac{L_\theta }{G}&0&\frac{Y_{\theta ,\phi}}{G}\\
X_{\theta ,\phi}G&0&R_\phi G&0\\
0&\frac{Y_{\theta ,\phi}}{G}&0& \frac{L_\phi }{G}\\
\end{array}
\right],
\end{align}
\normalsize	in which\small
	\begin{equation*}
\begin{aligned}
R_\theta&\triangleq\mathbf{a}_{\mathrm{R}}^{\mathrm{H}}\tilde{\mathbf{K}}^2_{\mathrm{R}}\mathbf{a}_{\mathrm{R}},\\
R_\phi &\triangleq\mathbf{a}_{\mathrm{R}}^{\mathrm{H}}\tilde{\mathbf{P}}^2_{\mathrm{R}}\mathbf{a}_{\mathrm{R}},\\
V_\theta &\triangleq\mathbf{a}_\mathrm{T}^{\mathrm{H}}\tilde{\mathbf{K}}_{\mathrm{T}}\mathbf{F}\mathbf{F}^\mathrm{H}\mathbf{a}_\mathrm{T},\\
V_\phi &\triangleq\mathbf{a}_\mathrm{T}^{\mathrm{H}}\tilde{\mathbf{P}}_{\mathrm{T}}\mathbf{F}\mathbf{F}^\mathrm{H}\mathbf{a}_\mathrm{T}.\\
\end{aligned}
\ \
\begin{aligned}
L_\theta &\triangleq GT_\theta -|V_\theta |^2,\\
L_\phi &\triangleq GT_\phi -|V_\phi |^2,\\
T_\theta &\triangleq\mathbf{a}_\mathrm{T}^{\mathrm{H}}\tilde{\mathbf{K}}_{\mathrm{T}}\mathbf{F}\mathbf{F}^\mathrm{H}\tilde{\mathbf{K}}_{\mathrm{T}}\mathbf{a}_\mathrm{T},\\
T_\phi &\triangleq\mathbf{a}_\mathrm{T}^{\mathrm{H}}\tilde{\mathbf{P}}_{\mathrm{T}}\mathbf{F}\mathbf{F}^\mathrm{H}\tilde{\mathbf{P}}_{\mathrm{T}}\mathbf{a}_\mathrm{T},\\	
\end{aligned}
\ \
\begin{aligned}
G&\triangleq\mathbf{a}_\mathrm{T}^{\mathrm{H}}\mathbf{F}\mathbf{F}^\mathrm{H}\mathbf{a}_\mathrm{T},\\
X_{\theta ,\phi}&\triangleq\mathbf{a}_{\mathrm{R}}^{\mathrm{H}}\tilde{\mathbf{K}}_{\mathrm{R}}\tilde{\mathbf{P}}_{\mathrm{R}}\mathbf{a}_{\mathrm{R}},\\
Y_{\theta ,\phi}&\triangleq GY'_{\theta ,\phi}-\Re\{V_\phi V_\theta ^*\},\\	
Y'_{\theta ,\phi}&\triangleq\Re\{\mathbf{a}_\mathrm{T}^{\mathrm{H}}\tilde{\mathbf{P}}_{\mathrm{T}}\mathbf{F}\mathbf{F}^\mathrm{H}\tilde{\mathbf{K}}_{\mathrm{T}}\mathbf{a}_\mathrm{T}\},\\	
\end{aligned}
\end{equation*}\normalsize
\textcolor{black}{Note that $G$ denotes the transmit array gain in a direction $\theta_\mathrm{T}$, while $R_\theta, R_\phi, T_\theta$, and $T_\phi$ are the information contents related to the spatial aspects of the received signal and correspond to $\theta_\mathrm{R}$, $\phi_\mathrm{R}$, $\theta_\mathrm{T}$, and $\phi_\mathrm{T}$, excluding the SNR, i.e., the integrands in \eqref{eq:FIM}. Similarly, $V_\theta, V_\phi, X_{\theta,\phi}$, and $Y'_{\theta,\phi}$ respectively represent the mutual spatial information between $\theta_\mathrm{T}$ and $\beta$, $\phi_\mathrm{T}$ and $\beta$,  $\theta_\mathrm{R}$ and $\phi_\mathrm{R}$, and $\theta_\mathrm{T}$ and $\phi_\mathrm{T}$. Consequently, $L_\theta$ and $L_\phi$ represent the \textit{equivalent} Fisher spatial information of $\theta_\mathrm{T}$ and $\phi_\mathrm{T}$, respectively, after removing the dependence on $\beta$. Finally, $Y_{\theta,\phi}$ denotes the equivalent mutual information of $\theta_\mathrm{T}$ and $\phi_\mathrm{T}$, after removing the dependence on $\beta$.}

The CRLB of the channel parameters for arbitrary array geometries is provided below.
\begin{proposition}\label{prop:CRB}
	Based on the FIM in (\ref{eq:partitioned}), the CRLBs of the DOA, DOD and TOA are given by
	\begin{subequations}\label{eq:closed_CRLB}
		\begin{align}
		\CRLB\left(\theta_{\mathrm{R}}\right)&\textcolor{black}{=\frac{1}{\gamma|\beta|^2~G\left(R_\theta -\frac{X_{\theta ,\phi}^2}{R_\phi}\right)}}=\frac{R_\phi}{\gamma\zeta_1|\beta|^2 G},\\
		\CRLB\left(\phi_{\mathrm{R}}\right)&\textcolor{black}{=\frac{1}{\gamma|\beta|^2~G\left(R_\phi -\frac{X_{\theta ,\phi}^2}{R_\theta}\right)}}=\frac{R_\theta}{\gamma\zeta_1|\beta|^2 G},\\
		\CRLB\left(\theta_{\mathrm{T}}\right)&\textcolor{black}{=\frac{G}{\gamma|\beta|^2\left(L_\theta -\frac{Y_{\theta ,\phi}^2}{L_\phi}\right)}}=\frac{GL_\phi}{\gamma\zeta_2|\beta|^2 },\\
		\CRLB\left(\phi_{\mathrm{T}}\right)&\textcolor{black}{=\frac{G}{\gamma|\beta|^2\left(L_\phi -\frac{Y_{\theta ,\phi}^2}{L_\theta}\right)}}=\frac{GL_\theta}{\gamma\zeta_2|\beta|^2 },\\
		\CRLB(\tau)&=\frac{1}{(4\pi^2W_{\mathrm{eff}}^2)(\gamma|\beta|^2G)}\label{eq:sing_crlb_tau}.
		\end{align}
	\end{subequations}
	where $\zeta_1=R_\theta R_\phi -X_{\theta ,\phi}^2$, and $\zeta_2=L_\theta L_\phi -Y_{\theta ,\phi}^2$.
\end{proposition}
\begin{proof}
	See Appendix \ref{app:CRLB}.
\end{proof}

\textcolor{black}{Note that the CRLBs consists of two components: Firstly, there is an SNR component represented by $\gamma|\beta|^2~G$ for the receiver angles, and by $G/(\gamma|\beta|^2)$ for the transmitter angles. This component is inversely proportional to the CRLBs, which means with higher SNR, the CRLBs of the channel parameters decrease. Secondly, there is a spatial information part in the parentheses containing the equivalent information after removing the dependence on the other parameter. }

Recall that for 2D localization, when the UE and BS are located in the $xy$-plane. Thus, ignoring the terms relating to the coupling between $\theta$ and $\phi$ in Proposition 1 leads to
 \begin{subequations}\label{eq:closed_CRLB2D}
	\begin{align}
	\CRLB(\phi_{\mathrm{R}})&=\frac{1}{\gamma|\beta|^2 R_\phi G},\qquad\qquad
	\CRLB(\phi_\mathrm{T})=\frac{G}{\gamma|\beta|^2L_\phi },
	\end{align}
\end{subequations}
while $\CRLB(\tau)$ is unchanged.

\section{Fisher Information of the Location Parameters}\label{sec:PEB_OEB}
In the preceding sections, we have seen how the FIM of the multipath channel parameters can be approximated by multiple single-path FIMs. We have also derived the single-path FIM for different settings. In this section, we derive the PEB and the OEB by applying a  transformation \cite{kay1993} to the EFIM of DOA, DOD, and TOA, computed from \eqref{eq:matrix_I}, to obtain the exact FIM of position and orientation. We also transform $\mathbf{J}_{{\boldsymbol{\varphi}}_\mathrm{ch}}$,  defined in (\ref{eq:partitioned}),  to obtain the approximate one.

\subsection{PEB and OEB: Exact Approach}
\subsubsection{General Formulation}\label{sec:general_formulation}

In this section, we derive the PEB and OEB based on the EFIM of the multipath channel parameters of interest $\boldsymbol{\varphi}_{\mathrm{CH}}\triangleq[\boldsymbol{\theta}^\mathrm{T},\boldsymbol{\phi}^\mathrm{T},\boldsymbol\tau^{\mathrm{T}}]^\mathrm{T}$. We do so by first transforming  $\mathbf{J}^\mathrm{e}_{\boldsymbol\varphi_\mathrm{CH}}$ to a FIM of the location parameters $\boldsymbol{\varphi}_{\mathrm{L}}\triangleq[\mathbf{o}^\mathrm{T},\mathbf{p}^\mathrm{T},\mathbf{q}^\mathrm{T}]^\mathrm{T}$, where
 \begin{align}\notag
	\boldsymbol{\theta}=\left[\boldsymbol{\theta}_{\mathrm{R}}^\mathrm{T},\boldsymbol{\theta}_\mathrm{T}^\mathrm{T}\right]^\mathrm{T},\qquad
	\boldsymbol{\phi}=\left[\boldsymbol{\phi}_{\mathrm{R}}^\mathrm{T},\boldsymbol{\phi}_\mathrm{T}^\mathrm{T}\right]^\mathrm{T}, \qquad
	\mathbf{q}\triangleq\left[\mathbf{q}_2^\mathrm{T},\mathbf{q}_3^\mathrm{T},\cdots,\mathbf{q}_M^\mathrm{T}\right]^\mathrm{T}.
	\end{align}
Towards that, we write
\begin{align}
\mathbf{J}_{\boldsymbol{\varphi}_{\mathrm{L}}}&\triangleq\boldsymbol{\Upsilon}\mathbf{J}^\mathrm{e}_{\boldsymbol{\varphi}_{\mathrm{CH}}}\boldsymbol{\Upsilon}^\mathrm{T}\triangleq\begin{bmatrix}
\boldsymbol\Psi&\boldsymbol\Phi\\
\boldsymbol\Phi^\mathrm{T}&\boldsymbol\Omega
\end{bmatrix},\label{eq:tranformation}
\end{align}
where $\boldsymbol{\Psi}\in\mathbb{R}^{5\times{5}}$, $\boldsymbol{\Omega}\in\mathbb{R}^{3(M-1)\times{3(M-1)}}$ and
	\begin{align}
\boldsymbol{\Upsilon}&\triangleq\frac{\partial\boldsymbol{\varphi}^\mathrm{T}_\mathrm{CH}}{\partial\boldsymbol{\varphi}_\mathrm{L}}=\begin{bmatrix}
\frac{\partial\boldsymbol{\theta}^\mathrm{T}}{\partial\mathbf{o}}& \frac{\partial\boldsymbol{\phi}^\mathrm{T}}{\partial\mathbf{o}}&\frac{\partial{\boldsymbol\tau^{\mathrm{T}}}}{\partial\mathbf{o}}\\
\frac{\partial\boldsymbol{\theta}^\mathrm{T}}{\partial\mathbf{p}}& \frac{\partial\boldsymbol{\phi}^\mathrm{T}}{\partial\mathbf{p}}&\frac{\partial{\boldsymbol\tau^{\mathrm{T}}}}{\partial\mathbf{p}}\\
\frac{\partial\boldsymbol{\theta}^\mathrm{T}}{\partial\mathbf{q}}& \frac{\partial\boldsymbol{\phi}^\mathrm{T}}{\partial\mathbf{q}}&\frac{\partial{\boldsymbol\tau^{\mathrm{T}}}}{\partial\mathbf{q}}\\
\end{bmatrix}.\label{eq:matrix_t}
\end{align}	
Consequently, the EFIM of the position $\mathbf{p}$ and orientation $\mathbf{o}$ is found via Schur's complement as
\begin{align}
{\mathbf{J}}^\mathrm{e}_{\mathbf{o,p}}=\boldsymbol\Psi-\boldsymbol\Phi\boldsymbol\Omega^{-1}\boldsymbol\Phi^\mathrm{T}.
\end{align}
Finally, the squared-PEB (SPEB) and squared-OEB (SOEB) are defined as
\begin{subequations}
	\begin{align}
\mathrm{SOEB}&=\left[({{\mathbf{J}}^\mathrm{e}_{\mathbf{o,p}}})^{-1}\right]_{1,1} + \left[({{\mathbf{J}}^\mathrm{e}_{\mathbf{o,p}}})^{-1}\right]_{2,2} \label{eq:OEB_def},\\
\mathrm{SPEB}&=\left[({{\mathbf{J}}^\mathrm{e}_{\mathbf{o,p}}})^{-1}\right]_{3,3}+\left[({{\mathbf{J}}^\mathrm{e}_{\mathbf{o,p}}})^{-1}\right]_{4,4}+\left[({{\mathbf{J}}^\mathrm{e}_{\mathbf{o,p}}})^{-1}\right]_{5,5}.\label{eq:PEB_def}
\end{align}
\end{subequations}
\subsubsection{Transformation for Uplink and Downlink}\label{sec:rotation_matrix}
The relationships governing the UE position and orientation with the BS and UE angles are different. Therefore, unlike $\mathbf{J}_{\boldsymbol{\varphi}}$, the structure of $\boldsymbol{\Upsilon}$ and, effectively, ${\mathbf{J}}^\mathrm{e}_{\mathbf{o,p}}$, depends on whether the uplink or downlink is used for signal transmission. For this reason, we switch to the explicit notation with the subscripts $\mathrm{BS}$ and $\mathrm{UE}$,
\begin{align}
\boldsymbol{\theta}=\begin{cases}
\left[\boldsymbol{\theta}_\mathrm{BS}^\mathrm{T},\boldsymbol{\theta}_\mathrm{UE}^\mathrm{T}\right]^\mathrm{T}, &\mathrm{uplink}\\
\left[\boldsymbol{\theta}_\mathrm{UE}^\mathrm{T},\boldsymbol{\theta}_\mathrm{BS}^\mathrm{T}\right]^\mathrm{T},&\mathrm{downlink}\\
\end{cases},\qquad\qquad
\boldsymbol{\phi}=\begin{cases}
\left[\boldsymbol{\phi}_\mathrm{BS}^\mathrm{T},\boldsymbol{\phi}_\mathrm{UE}^\mathrm{T}\right]^\mathrm{T}, &\mathrm{uplink}\\
\left[\boldsymbol{\phi}_\mathrm{UE}^\mathrm{T},\boldsymbol{\phi}_\mathrm{BS}^\mathrm{T}\right]^\mathrm{T},&\mathrm{downlink}.
\end{cases}
\end{align}
where $\boldsymbol{\phi}_\mathrm{BS}$ and $\boldsymbol{\theta}_\mathrm{BS}$ denote the vectors of the azimuth and elevation angles at the BS, and $\boldsymbol{\phi}_\mathrm{UE}$ and $\boldsymbol{\theta}_\mathrm{UE}$ are the azimuth and elevation angles  at the UE.
\begin{figure}[!t]
		\vspace{-2mm}
		\centering
	\includegraphics[scale=0.8]{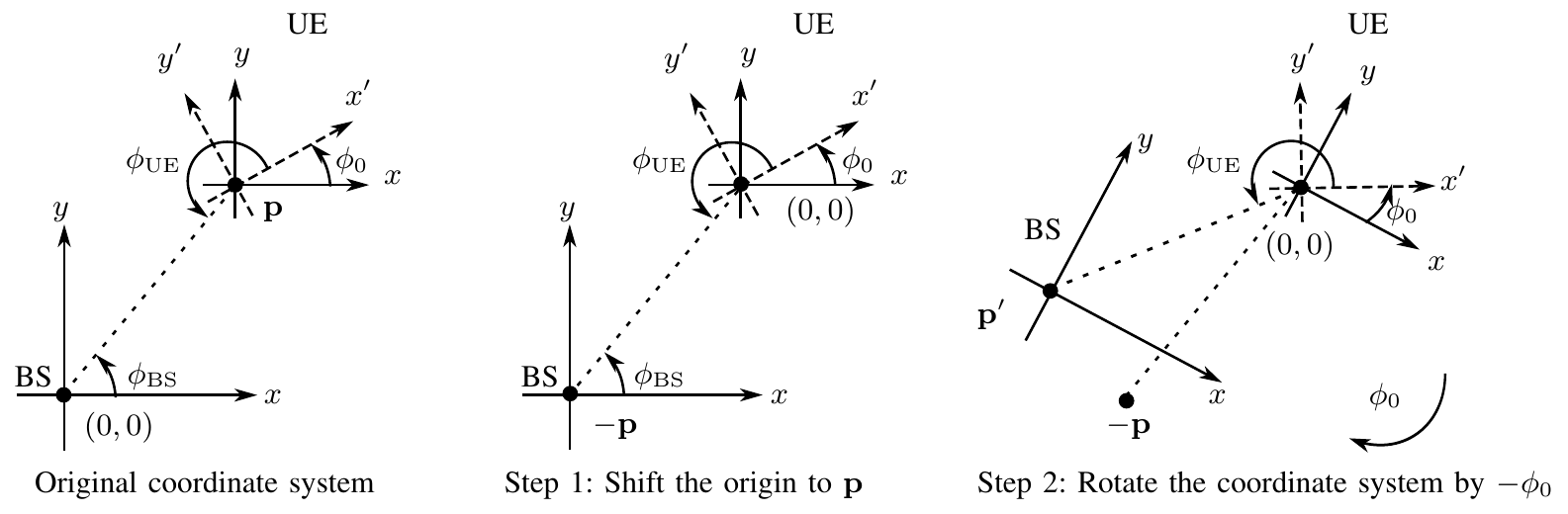}
	\vspace{-2mm}
	\caption{Two-step derivation of the UE angle in 2D. It is easy to see that $\phi_\mathrm{UE}=\tan^{-1}\left({p_y'}/{p_x'}\right)$, where $\mathbf{p}'=-{\mathbf{{R}}}_z(-\phi_0)\mathbf{p}$.}
	\label{fig:rotation_geo}
	\vspace{-8mm}
\end{figure}

\textcolor{black}{Starting with LOS and using the spherical coordinates, it can be seen from Fig.~\ref{fig:scenario} that
\begin{subequations}\label{LOS_geometry_BS}
	\begin{align}
	\theta_{\mathrm{BS},1}=\cos^{-1}\left({p_z}/{\|\mathbf{p}\|}\right),\qquad
	\phi_{\mathrm{BS},1}=\tan^{-1}\left({p_y/p_x}\right),\qquad
	\tau_1={\|\mathbf{p}\|}/{c}.
	\end{align}
\end{subequations}
However, the relationship of the UE angles with the position and orientation angles are not as obvious. Therefore, we resort to the two-step procedure illustrated in Fig.~\ref{fig:rotation_geo} for 2D, but easily extensible to 3D. In the first step, we shift the coordinate system origin to the UE, hence, the BS is shifted to $-\mathbf{p}$. In the second step, the coordinate system is rotated in the negative direction of the orientation angle ($\phi_0$). Consequently, the BS location is also rotated, and the UE angles are then taken as the spherical coordinates of the new BS location. Mathematically, this location is given by $
\mathbf{p}'=-\mathbf{R}_z(-\phi_0)\mathbf{p}=-\mathbf{R}_z^{-1}(\phi_0)\mathbf{p},$
where $\mathbf{R}_z(\phi_0)$ is the rotation matrix in the direction $\phi_0$ around the positive $z$-axis. Generalizing this result to the 3D case yields, 
\begin{align}
\mathbf{p}'=-\mathbf{R}^{-1}(\theta_0,\phi_0)\mathbf{p}.
\end{align}
Consequently, defining $\mathbf{p}'\triangleq[p_x',p_y',p_z']^\mathrm{T}$ and noting that $\|\mathbf{p}\|=\|\mathbf{p}'\|$, we write
\begin{subequations}\label{LOS_geometry_UE}
	\begin{align}
	\theta_{\mathrm{UE},1}=\cos^{-1}\left({p'_z}/{\|\mathbf{p}\|}\right),\qquad	\phi_{\mathrm{UE},1}=\tan^{-1}\left({p_y'/p'_x}\right).
	\end{align}
\end{subequations}
With the right-hand rule in mind, the rotation considered in this paper (See Fig.~\ref{fig:scenario}) is a rotation by $\phi_0$ around the $z$-axis, followed by another rotation by $\theta_0$ around the {\textit{negative}} $x'$-axis. Thus, the rotation matrix is given by \cite{vince2011rotation} 
\begin{align}\label{eq:rotation_matrix}
\mathbf{R}(\theta_0,\phi_0)=\mathbf{R}_z(\phi_0)\mathbf{R}_{-x'}(\theta_0)=\begin{bmatrix}
\cos\phi_0&-\sin\phi_0\cos\theta_0&-\sin\phi_0\sin\theta_0\\
\sin\phi_0&\cos\phi_0\cos\theta_0&\cos\phi_0\sin\theta_0\\
0&-\sin\theta_0&\cos\theta_0
\end{bmatrix}.
\end{align}
Note that $\mathbf{R}(\theta_0,\phi_0)$ is orthogonal and hence satisfies $\mathbf{R}^{-1}(\theta_0,\phi_0)=\mathbf{R}^\mathrm{T}(\theta_0,\phi_0)$.}

\textcolor{black}{Next, considering the NLOS paths ($2\leq{m}\leq{M}$) and using the same procedure, the following relations can be obtained
	\begin{align}\label{NLOS_geometry}
	\begin{aligned}	\theta_{\mathrm{UE},m}&=\cos^{-1}\left({w'_{m,z}}/{\|\mathbf{w}_m\|}\right),\\
	\phi_{\mathrm{UE},m}&=\tan^{-1}\left({w'_{m,y}/w'_{m,x}}\right),\\
		\tau_m&=\left(\|\mathbf{q}_m\|+\|\mathbf{w}_m\|\right)/{c},
	\end{aligned}\qquad
	\begin{aligned}		\theta_{\mathrm{BS},m}&=\cos^{-1}\left({q_{m,z}}/{\|\mathbf{q}_m\|}\right),\\
	\phi_{\mathrm{BS},m}&=\tan^{-1}\left({q_{m,y}/q_{m,x}}\right),\\
	\mathbf{w}_m&=\mathbf{p}-\mathbf{q}_m,
\end{aligned}
	\end{align}
where  $\mathbf{w}'_m\triangleq[w'_{m,x},w'_{m,y},w'_{m,z}]^\mathrm{T}=-\mathbf{R}^\mathrm{T}(\theta_0,\phi_0)\mathbf{w}_m$. Based on \eqref{LOS_geometry_BS}, \eqref{LOS_geometry_UE}, and \eqref{NLOS_geometry}, the non-zero elements of $\boldsymbol{\Upsilon}$ are listed in Appendix \ref{app:derivatives}.}

\subsection{PEB and OEB: Approximate Approach}
In Section \ref{sec:approx_crlb}, it was concluded that, under certain conditions, the multiple paths arriving at the receiver can be treated as non-interfering paths carrying independent information. Thus, we can write the total FIM of position and orientation as a sum of the individual FIMs obtained by transforming the CRLBs of DODs, DOAs, and TOAs
of these paths individually.

\begin{proposition}
Let $\mathbf{J}^{(m)}_{\boldsymbol{\varphi}_s}$ be the FIM of the m$^{\mathrm{th}}$ path parameters depicted in Fig. \ref{fig:FIM_reorder} (left), so that
\begin{align}\label{eq:FIM_reordered}
\mathbf{J}^{(m)}_{\boldsymbol{\varphi}_s}\triangleq\begin{bmatrix}
\boldsymbol{\Lambda}_m&\boldsymbol{\Pi}_m\\
\boldsymbol{\Pi}^\mathrm{T}_m&\boldsymbol{\Xi}_m\\
\end{bmatrix},
\end{align}
where $\boldsymbol{\Lambda}_m\in\mathbb{R}^{5\times{5}}$ is the FIM of $\theta_{\mathrm{R},m},\theta_{\mathrm{T},m},\phi_{\mathrm{R},m},\phi_{\mathrm{T},m},$ and $\tau_{m}$, $\boldsymbol{\Xi}_m\in\mathbb{R}^{2\times{2}}$ is the FIM of $\beta_{\mathrm{R},m}$ and $\beta_{\mathrm{I},m}$, and \textcolor{black}{ $\boldsymbol{\Pi}_m\in\mathbb{R}^{5\times{2}}$ is the mutual information matrix of parameters of $\boldsymbol{\Lambda}_m$ and $\boldsymbol{\Xi}_m$ }. Moreover, denote the EFIM of the m$^{\mathrm{th}}$ DOA, DOD, and TOA by
\begin{align}
\boldsymbol{\Lambda}^\mathrm{e}_m=\boldsymbol{\Lambda}_m-\boldsymbol{\Pi}_m\boldsymbol{\Xi}^{-1}_m\boldsymbol{\Pi}^\mathrm{T}_m,
\end{align}
and the corresponding transformation matrix in block form by
\begin{align}
\boldsymbol{\Upsilon}_m\triangleq\begin{cases}
\overline{\boldsymbol{\Upsilon}}_1,&m=1\\
\begin{bmatrix}
\overline{\boldsymbol{\Upsilon}}_m^\mathrm{T}&\overline{\overline{\boldsymbol{\Upsilon}}}_m^\mathrm{T}\\
\end{bmatrix}^\mathrm{T}, &2\leq m\leq{M}
\end{cases}
\end{align}
where $\overline{\boldsymbol{\Upsilon}}_m$ is the $5\times{5}$ matrix relating to $\mathbf{o}$ and $\mathbf{p}$, and $\overline{\overline{\boldsymbol{\Upsilon}}}_m$ is the $3\times{5}$ matrix relating to $\mathbf{q}_m$.
Then, the approximate EFIM of $\mathbf{o}$ and $\mathbf{p}$ is given by
\begin{align}
\tilde{\mathbf{J}}^{\mathrm{e}}_{\mathbf{o,p}}\triangleq\sum_{m=1}^{M}\mathbf{J}_{\mathbf{o,p}}^{(m)}=&\sum_{m=1}^{M}\overline{\boldsymbol{\Upsilon}}_m\boldsymbol{\Lambda}_m\overline{\boldsymbol{\Upsilon}}_m^\mathrm{T}\notag\\-&\underbrace{\sum_{m=1}^{M}\overline{\boldsymbol{\Upsilon}}_m\boldsymbol{\Pi}_m\boldsymbol{\Xi}^{-1}_m\boldsymbol{\Pi}^\mathrm{T}_m\overline{\boldsymbol{\Upsilon}}_m^\mathrm{T}}_{\mathrm{path\ gains\ uncertainty}}-\underbrace{\sum_{m=2}^{M}\overline{\boldsymbol{\Upsilon}}_m\boldsymbol{\Lambda}^\mathrm{e}_m\overline{\overline{\boldsymbol{\Upsilon}}}_m^\mathrm{T}\left(\overline{\overline{\boldsymbol{\Upsilon}}}\boldsymbol{\Lambda}^\mathrm{e}_m\overline{\overline{\boldsymbol{\Upsilon}}}^\mathrm{T}\right)^{-1}\overline{\overline{\boldsymbol{\Upsilon}}}_m\boldsymbol{\Lambda}^\mathrm{e}_m\overline{\boldsymbol{\Upsilon}}_m^\mathrm{T}}_{\mathrm{clusters\ locations\ uncertainty}}.\label{eq:Proposition2}
\end{align}
\end{proposition}
\begin{proof}
	See Appendix \ref{app:OP_FIM}.
\end{proof}
We make the following remarks form \eqref{eq:Proposition2}. Firstly, due to the additive nature of the EFIM, the FIM of the useful localization information (TOA, DOA, DOD) of the $M$ paths accumulate positively to construct the first term. On the other hand, the channel gain, $\beta_m$, is a nuisance unknown parameter which needs to be estimated despite not being useful for localization. Not knowing $\beta_m$ decreases the amount of available information as highlighted by the negative second term comprising $\boldsymbol{\Pi}_m$, the mutual information relating $\beta$ with TOA, DOA, and DOD. Finally, since $m=1$ is assumed to be a LOS path, the third term is defined starting from $m=2$. This term is also negative, indicating the information loss due to the unknown clusters' locations, $\mathbf{q}_m$.

\subsection{Closed-Form Expressions for LOS}\label{sec:OEB_PEB_closed_form}
Although it is hard to derive closed-form solutions of the general case of PEB and OEB, here we present expressions for the LOS case ($M=1$). 
\begin{proposition}
\textcolor{black}{For the problem set in Section \ref{sec:prob_form}, in the existence of a LOS path only, the {3D} localization PEB and OEB of a UE located at $\mathbf{p}$ with an orientation angle $\mathbf{o}$ are given by
		\begin{subequations}\label{eq:SOEB_SPEB_3D}
	\begin{align}
	\mathrm{SPEB}=&\|\mathbf{p}\|^2\CRLB(\theta_\mathrm{BS})+\|\mathbf{p}\|^2\sin^2\theta_\mathrm{BS}\CRLB(\phi_\mathrm{BS})+c^2\CRLB(\tau),\\
	\mathrm{SOEB}=&b_1\CRLB(\theta_\mathrm{BS})+b_2\CRLB(\phi_\mathrm{BS})+b_3\sigma^2_{\theta_\mathrm{BS}\phi_\mathrm{BS}}\notag\\
	+&b_4\CRLB(\theta_\mathrm{UE})+b_5\CRLB(\phi_\mathrm{UE})+b_6\sigma^2_{\theta_\mathrm{UE}\phi_\mathrm{UE}}.
	\end{align}
\end{subequations}
where $\sigma^2_{\theta_\mathrm{BS}\phi_\mathrm{BS}}$ and $\sigma^2_{\theta_\mathrm{UE}\phi_\mathrm{UE}}$ are covariance terms arising from the mutual information of the angles in the subscript, and $b_1,...,b_6$ are as given in Appendix \ref{app:LOS_PEB_OEB}.} 
\end{proposition}
\begin{proof}
	See Appendix \ref{app:LOS_PEB_OEB}.
\end{proof}
\textcolor{black}{In light of \eqref{eq:SOEB_SPEB_3D}, it can be seen that SPEB depends on the BS angles rather than the UE angles. In other words, SPEB depends on CRLB(DOA) in the uplink, and CRLB(DOD) in the downlink, which have different expressions in \eqref{eq:closed_CRLB}. Thus, SPEB is asymmetric in these two cases. On the other hand, SOEB depends on both UE angles and BS angles, albeit with different weights in the uplink and downlink.  As a result, although the SOEB expression in \eqref{eq:SOEB_SPEB_3D} is valid for both uplink and downlink, SOEB is asymmetric in general. Finally, for the 2D special case, it can be shown that, discarding the terms related to the elevation angles, $\mathrm{SOEB}=\CRLB(\phi_\mathrm{BS})+\CRLB(\phi_\mathrm{UE})$ and $\mathrm{SPEB}=c^2\CRLB(\tau)+\|\mathbf{p}\|^2\CRLB(\phi_\mathrm{BS})$.}
\begin{table*}[!t]
	\small
	\begin{center}
		\setlength\extrarowheight{4pt}
		\caption{Scaling Factors of CRLBs of the Channel Parameters, PEB and OEB}
		\vspace{-8mm}
		\label{tab:scale}
		\begin{tabular}{|c|c|c|c|c|c||c|c|c|}
			\hline
			&CRLB$(\theta_{\mathrm{R}})$&CRLB$(\theta_\mathrm{T})$&CRLB$(\phi_{\mathrm{R}})$ &CRLB$(\phi_\mathrm{T})$&CRLB$(\tau)$&$\mathrm{SOEB}$&$\mathrm{SPEB~(DL)}$&$\mathrm{SPEB~(UL)}$\\
			\hline
			URA&$N_{\mathrm{R}}^{-2}$&$N_{\mathrm{R}}^{-1}$&$N_{\mathrm{R}}^{-2}$&$N_{\mathrm{R}}^{-1}$&$N_{\mathrm{R}}^{-1}$&$N_{\mathrm{R}}^{-1}+N_{\mathrm{R}}^{-2}$&$N_{\mathrm{R}}^{-1}$&$N_{\mathrm{R}}^{-1}+N_{\mathrm{R}}^{-2}$\\			ULA &N/A&N/A&$N_{\mathrm{R}}^{-3}$&$N_{\mathrm{R}}^{-1}$&$N_{\mathrm{R}}^{-1}$&$N_{\mathrm{R}}^{-1}+N_{\mathrm{R}}^{-3}$&$N_{\mathrm{R}}^{-1}$&$N_{\mathrm{R}}^{-1}+N_{\mathrm{R}}^{-3}$\\
			\hline
		\end{tabular}
	\end{center}
	\vspace{-12mm}
\end{table*}

These results can be used to determine scaling laws. Evaluating \eqref{eq:closed_CRLB} for URA, the scaling laws for the CRLBs of the channel parameters are listed in Table \ref{tab:scale}, while those for ULA are obtained in \cite{Zohair2016}. We see that URAs and ULAs have different scaling, in that for $\CRLB(\phi_\mathrm{R})$ scales with $1/N^2_\mathrm{R}$ for URAs, but with $1/N^3_\mathrm{R}$ for ULAs. This can be explained by noting that these scaling factor consist of two multiplicative components: SNR improvement that scales with $1/N_\mathrm{R}$ for both geometries, and a spatial resolution that depends on the squared number of antennas in the direction of $x$-axis, that is $1/N^2_\mathrm{R}$ for ULAs, and $1/(\sqrt{N_\mathrm{R}})^2$ for URAs.

\textcolor{black}{To obtain scaling laws for SPEB and SOEB, we highlight that while the expressions in \eqref{eq:SOEB_SPEB_3D} are valid for both uplink and downlink, the values of the CRLBs in the expressions are different in these two cases (see \eqref{eq:closed_CRLB}). The CRLB(DOA) and CRLB(DOD) have different scaling laws and thus SPEB will scale differently in the uplink and downlink. However, since $b_1,...,b_6$ do not depend on the number of antennas, the scaling factors of SOEB is unchanged in both cases}. 

\section{Numerical Results and Discussion}\label{sec:sim}
\subsection{Simulation Environment}
Although the theoretical results are valid for any arbitrary array geometry, we focus on URAs, as an example of 3D localization. Particularly, we consider a scenario where a BS with square array is located in the $xz$-plane centered at the \textit{origin} with $\sqrt{N_{\mathrm{BS}}}\times \sqrt{N_{\mathrm{BS}}}$ antenna elements and a height of \textcolor{black}{$h_\mathrm{BS}=10$} meters. The UE, operating at $f=38$ GHz, is equipped with a square array which have $\sqrt{N_{\mathrm{UE}}}\times \sqrt{N_{\mathrm{UE}}}$ antenna elements, and assumed to be \textcolor{black}{tilted by an orientation angle of $0^\circ$ or $10^\circ$ in both azimuth and elevation}. We investigate the performance over a flat 120$^\circ$ sector of a sectorized cell with a radius of 50 meters. The UE is assumed to be located anywhere in this sector, which lies in the plane \textcolor{black}{$z=-h_\mathrm{BS}=-10$} meters. Moreover, we consider an ideal $\sinc$ pulse so that $W^2_\mathrm{eff}={W^2/3}$, where $W=125$ MHz, $E_\mathrm{s}/T_\mathrm{s}=0$ dBm, $N_0=-170$ dBm/Hz, and $N_{\mathrm{s}}=16$ pilot symbols. \textcolor{black}{ The LOS SNR at any location in the sector is given by $\mathrm{SNR}[dB]=144.24+20\log_{10} |\beta|+20\log_{10}\|\mathbf{a}_T\mathbf{F}\|$, with 95\% of the locations having an SNR of at least 30 dB}. We utilize the directional beamforming scheme defined in (\ref{eq:BF_Definition}). In the downlink case, the directions of the beams are \textcolor{black}{fixed and} chosen such that the beams centers are equispaced on the ground. On the other hand, in the uplink, \textcolor{black}{the centers of beams are fixed and equispaced on a virtual sector containing the BS. Initially, when the UE has zero orientation, i.e, lying in the $xz$-plane and facing the BS (See Fig.~\ref{fig:scenario}), this virtual sector lies in the horizontal plane $z=0$ meters. The beamforming angles are measured with respect to the UE array plane. Thus, when the orientation of the UE is non-zero, the virtual plane is rotated by the same orientation angles.} Fig.~\ref{fig:sectored_cell} depicts an example with $N_{\mathrm{B}}=25$ for both downlink and uplink.
\begin{figure}[!t]
	\centering
	\includegraphics[scale=0.32]{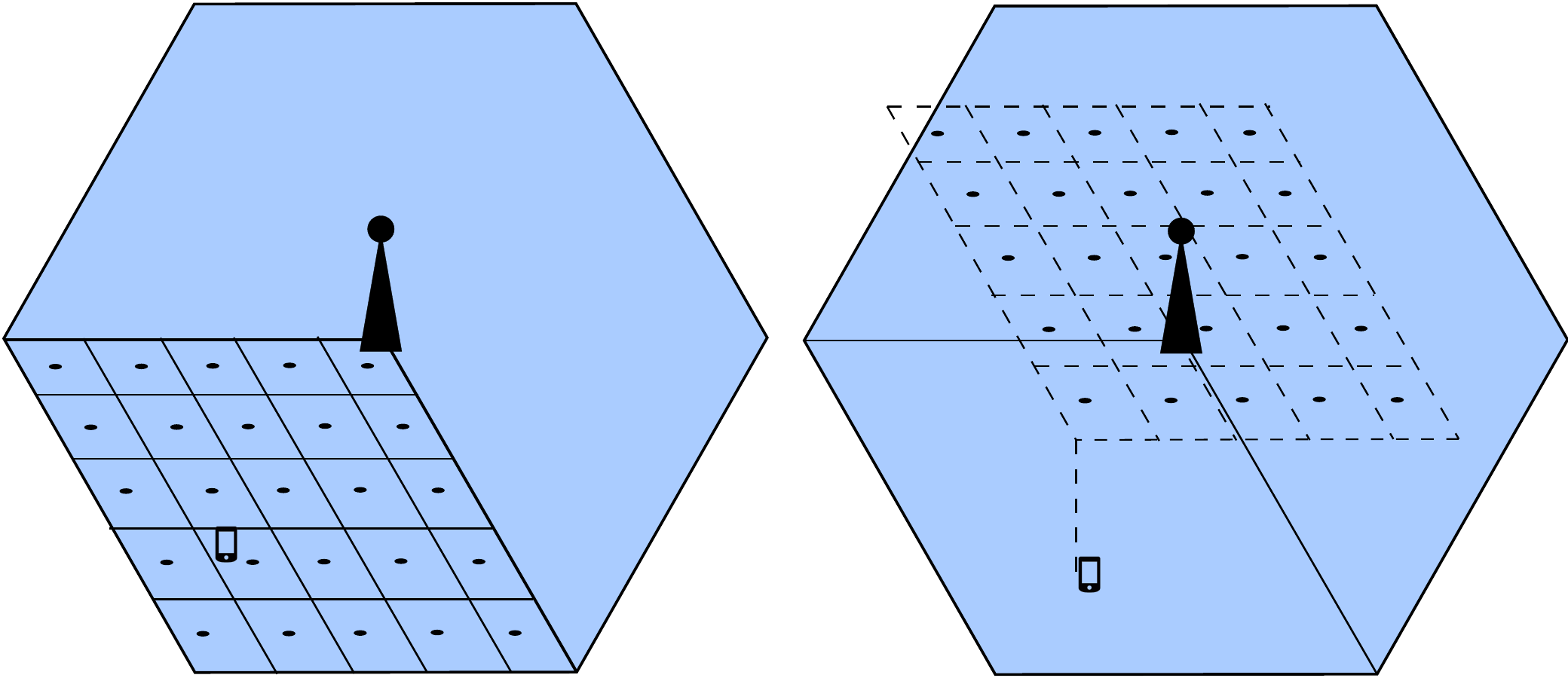}
			\vspace{-4mm}
	\caption{A cell sectorized into three sectors, each served by 25 beams directed towards a grid on the ground in the downlink (left) and towards a virtual grid in uplink (right). The grid has the same orientation as the UE. }\label{fig:sectored_cell}
		\vspace{-8mm}
\end{figure}

The environment comprises scatterers and reflectors, with scatterers distributed arbitrarily in the 3D space, and reflectors  placed close to the sector edge, \textcolor{black}{as shown in the example scenario in Fig.~\ref{fig:drawing}. We use 5 reflectors placed at the edge of the sector, not to obscure the area behind them if placed otherwise. We also use 15 scatterers distributed arbitrarily in the volume formed by the sector as base, and the BS as height. We only consider the clusters that contribute by a power greater than 10\% of the LOS power. It is seen during the simulations that this configuration leads to a maximum number of paths $M=6$ at any location in the studied sector, which is similar to the probability mass function in \cite{Akdeniz2014}}. Accordingly, defining $\vartheta_1=2\pi D_1/\lambda$ and $\vartheta_m = 2 \pi (d_{1,m}+d_{2,m})/\lambda$ for $m>1$, the complex channel gain of the $m^{\mathrm{th}}$ path is modeled by $\beta_m=|\beta_m|e^{j\vartheta _m}$ such that
\begin{align}\label{eq:beta_m}
|\beta_m|^2 &= \frac{\lambda^2}{(4\pi)^2}\begin{cases}
{1}/{D_1^2}& \text{LOS}\\
{\Gamma_{\mathrm{R}}}/{(d_{1,m}+d_{2,m})^2}& \text{reflector}\\
{\sigma^2_{\mathrm{RCS}}}/({4\pi}{( d_{1,m}d_{2,m})^2})& \text{scatterer},\\
\end{cases}
\end{align}
where $\sigma_{\mathrm{RCS}}^2=50$ m$^2$, and $\Gamma_{\mathrm{R}}=0.7$ are the radar cross section, and the reflection coefficient, respectively. \textcolor{black}{Although our derivations are valid for any path loss model, we use the model in \eqref{eq:beta_m} and the corresponding parameters values to get comparative insights into the role of reflectors and scatterers on the performance bounds. This may not be the typical case in reality where scatterers are characterized by the roughness of the surfaces, which would lead to random path loss, and consequently, random PEB and OEB.} The locations of reflectors are computed using the virtual transmitter method \cite{gentner2016multipath}.
We evaluate the PEB and OEB for the following scenarios:
\begin{figure}[!t]
	\centering
			\vspace{-2mm}
	\includegraphics[width=8cm,height=4.8cm]{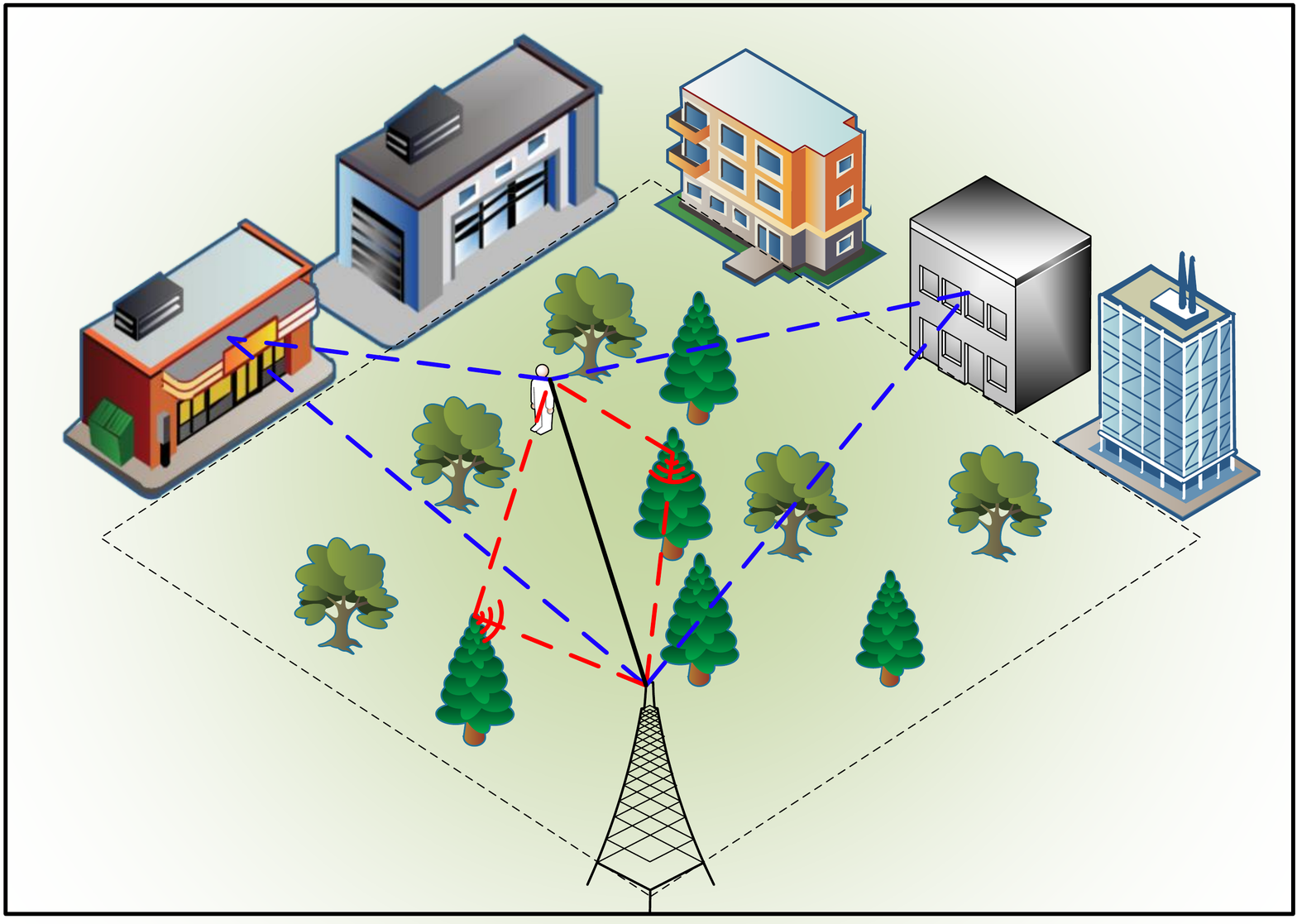}
		\vspace{-2mm}
	\caption{A scenario with LOS (black), 2 reflectors (blue) and 2 scatterers (red).}\label{fig:drawing}
	\vspace{-2mm}
\end{figure}
\begin{enumerate}
	\item \emph{LOS}: Free space propagation only, without NLOS paths.
	\item \emph{LOS+R}: A LOS path and $M-1$ reflected NLOS paths.
	\item \emph{LOS+S}: A LOS path and $M-1$ scattered NLOS paths.
	\item \emph{LOS+C}: A LOS path and a mix of both scattered and reflected NLOS paths.
	\item \emph{NLOS}: The LOS path is blocked, so only scattered and reflected NLOS paths exist.
\end{enumerate}
All the following results are obtained with $N_{\mathrm{B}}=25$, $N_\mathrm{T}=N_{\mathrm{R}}=144$, unless otherwise stated. \textcolor{black}{We choose equal array sizes at the UE and BS to make the comparison of uplink and downlink localization fair by having a symmetric channel setup. However, it is understood that more complexity as allowed at the BS and its array can grow to larger sizes such as that in \cite{Fan2017} to have up to 10,000 antennas, which will improve the localization bounds presented herein, subject to the number of beams used.}
\begin{figure}[!t]
	\centering
	\includegraphics{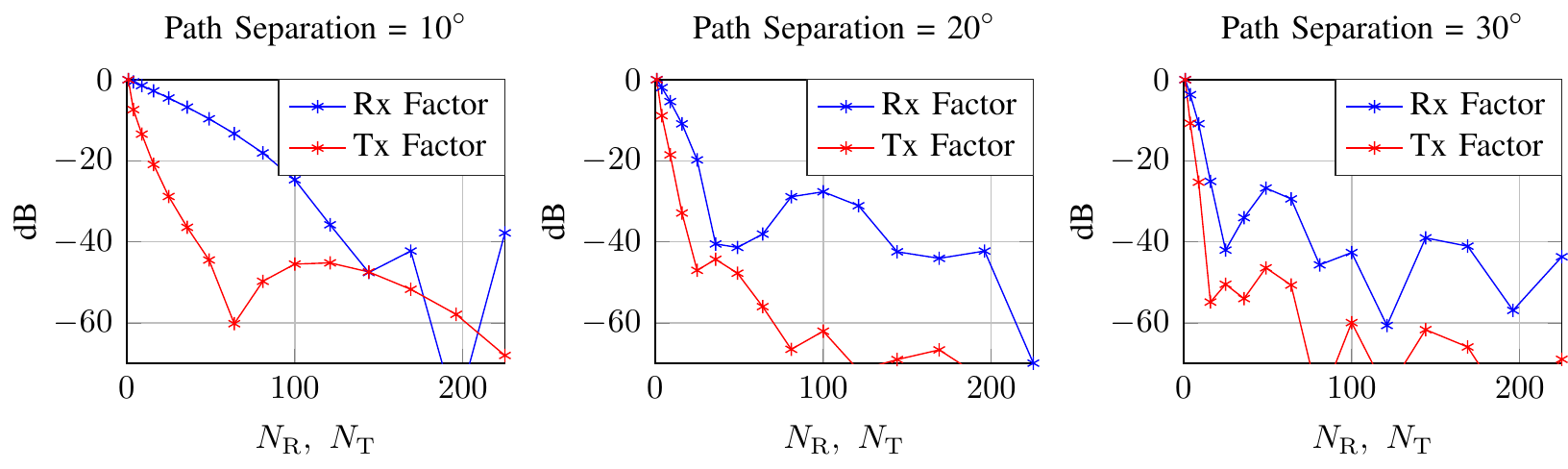}
	\vspace{-10mm}
	\caption{Receiver and Transmitter factors w.r.t. $N_\mathrm{R}$ and $N_\mathrm{T}$ for URA with different path separation angles. The separation angle is the angle difference between the azimuth and elevation angles of the two paths.}\label{fig:NtNrApprox}
	\vspace{-3mm}
\end{figure}

\textcolor{black}{\subsection{Tx and Rx Factors of the Approximate FIM}
We now investigate numerically the trend of the RX and TX factors in \eqref{eq:refcase} with respect to $N_\mathrm{R}$ and $N_\mathrm{T}$ and the path separation angles, as shown in Fig.~\ref{fig:NtNrApprox}. Each subfigure is obtained for a different separation angle, that is the azimuth and elevation angle difference between two paths. It can be seen that with a separation of $10^\circ$, $N_\mathrm{R}=100$ and $N_\mathrm{T}=16,$ the corresponding factor drops below -20 dB (1\% of the maximum). On the other hand, for higher separation angles, the two factors drop below -20 dB with less number of antennas. Finally, note that the approximate FIM is obtained by a combined effect of these two factors, plus the signal factor. Therefore, if $N_\mathrm{R}\geq{100}$ or $N_\mathrm{T}\geq{16}$, the total FIM will be almost diagonal.}

\begin{figure}[!t]
	\centering
	\includegraphics[trim={10mm 4mm 0cm 0mm},clip=true,scale=0.67]{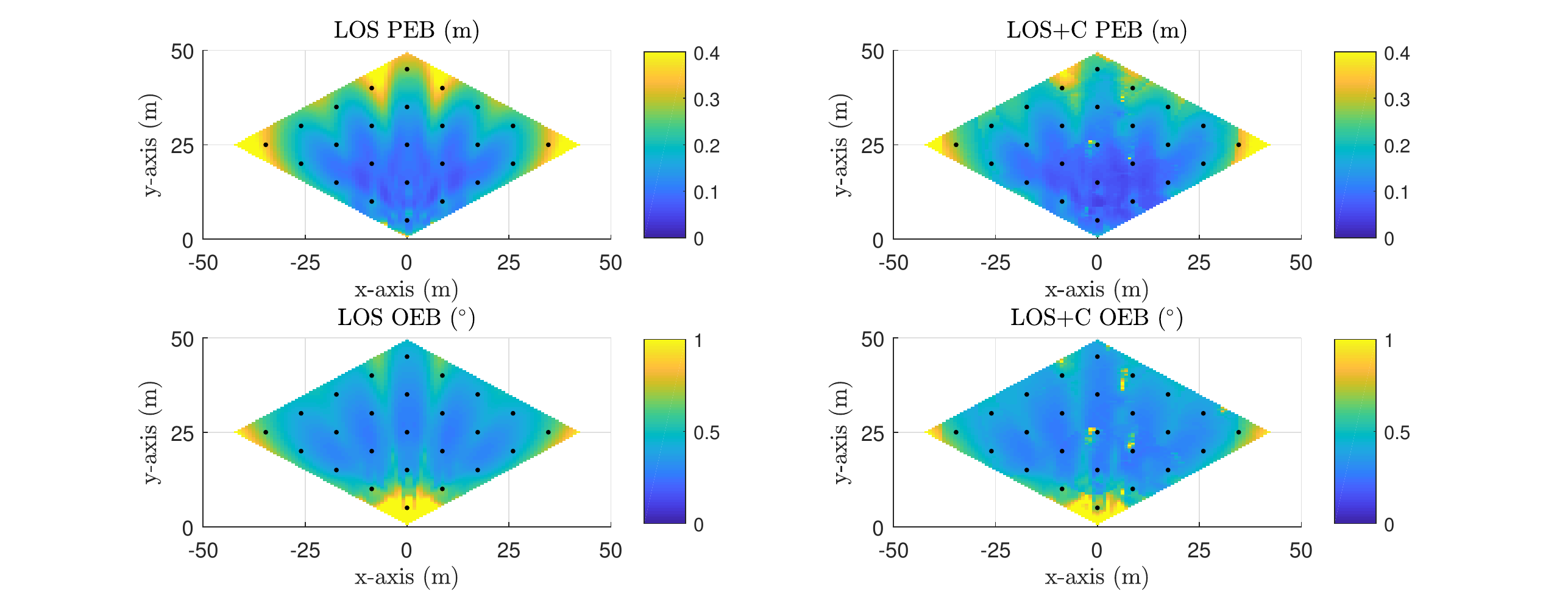}
	\caption{\textcolor{black}{PEB and OEB of downlink LOS (left) and LOS$+$C (right). $N_\mathrm{R}$=$N_\mathrm{T}$=144. The dots denote the centers of beams, $N_{\mathrm{B}}$=25.}}\label{fig:spatial1}
	\vspace{-10mm}
\end{figure}
\subsection{Downlink PEB and OEB}
Fig.~\ref{fig:spatial1} (left) shows the downlink PEB and OEB as a function of the UE location for the LOS case with the BS located at $(0,0)$. With $N_\mathrm{T}=N_{\mathrm{R}}=144$, and $N_{\mathrm{B}}=25$, the maximum PEB in the sector is $40$ cm, while the maximum OEB is $1^\circ$ in the LOS scenario. Note that from (\ref{eq:SOEB_SPEB_3D}), the PEB increases with $\|\mathbf{p}\|$. This explains the yellow area around the corners. Moreover, the closer the UE to the BS, i.e., as $\theta_\mathrm{BS}\rightarrow\pi$, singularities appear in the FIM, and the OEB tends to worsen, hence the yellow areas around the BS. Scatterers and reflectors are introduced in the 3D space, so that a maximum of 5 clusters contribute at any given location. Based on that, Fig.~\ref{fig:spatial1} (right) shows the PEB and OEB for the LOS+C case. Although incorporating NLOS clusters in the localization does not lower the maximum bound value, it improves the bounds at those locations where the clusters' signal are received. In the illustrated example, the clusters mainly affect the top and center areas of the sector. Finally, note the yellow dots in the central area of the PEB and OEB (LOS+C). These dots occur because at these locations, the scatterer blocks the LOS, violating the unique parameters assumption, and causing singularities in the FIM.

\begin{figure}[!t]
	\centering
		\vspace{-2mm}
	\footnotesize
	\includegraphics{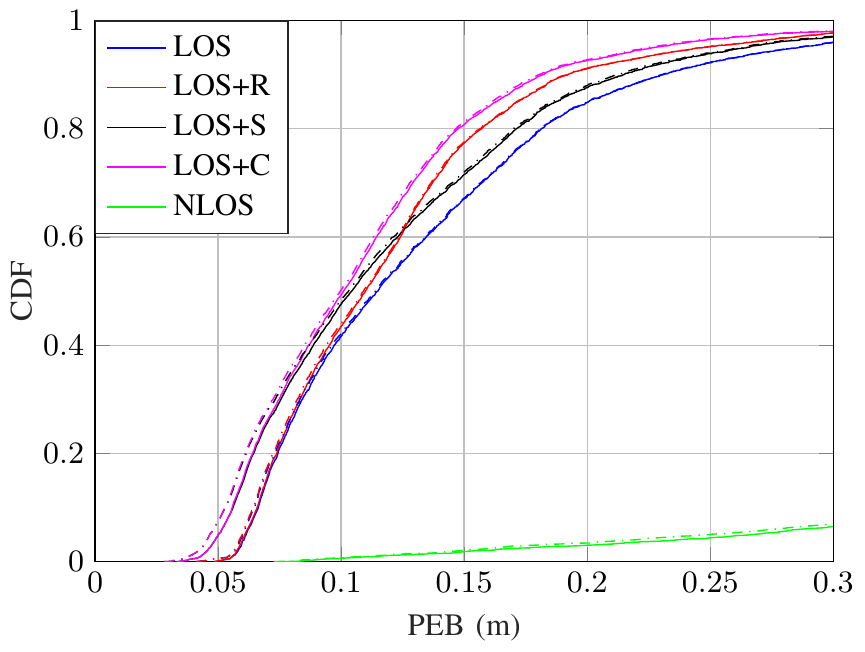}
		\vspace{-5mm}
	\caption{The CDF of downlink PEB for different scenarios using the exact (solid) and approximate (dashed) FIM approaches. }\label{fig:PEB_DL_scenarios}
			\vspace{-9mm}
\end{figure}

To obtain a more concise and quantitative assessment of the performance, we collect all the PEB and OEB values across the space and visualize them in a cumulative distribution function (CDF).  Subsequently, Fig.~\ref{fig:PEB_DL_scenarios} shows PEB obtained for all 5 considered scenarios. The PEB obtained from the approximate approach, is also shown in the figure. We observe the following: overall, scatterers and/or reflectors improve the localization performance, compared to the LOS-only scenario, despite the fact that more parameters need to be estimated. \textcolor{black}{Scatterers are mainly useful in providing rather low PEB improvement for many locations, while reflectors can provide modest PEB improvement for fewer locations.} When scatterers and reflectors are combined, we see both phenomena. It is also apparent that the approximate approaches closely follow the exact PEB and OEB and that the approximation always leads to a slightly lower PEB and OEB, due to the independent paths assumption, under this approach. Note that at a $90\%$ CDF, the PEB values for LOS, LOS+R, LOS+S, LOS+C are $0.23$ m, $0.21$ m, $0.19$ m and $0.18$ m, respectively. Moreover, note that the NLOS scenario is unreliable, with a PEB of $0.5$ m at a $13\%$ CDF, \textcolor{black}{and reaches 90\% CDF at a value that is irrelevant in mmWave localization. OEB curves (not shown) look similar to those in Fig.~\ref{fig:PEB_DL_scenarios}, with 90\% CDF ranging between $0.42^\circ$ and $0.5^\circ$ when a LOS exists. Finally, we obtained similar qualitative results with $N_\mathrm{R}=25$, but with PEB and OEB of $45-55$ cm and   $1.77^\circ-1.84^\circ$, at 90\% CDF respectively, when a LOS exists.}


\subsection{The Selection of $N_{\mathrm{B}}$}
\begin{figure}[!t]
	\centering
	\footnotesize
	\includegraphics{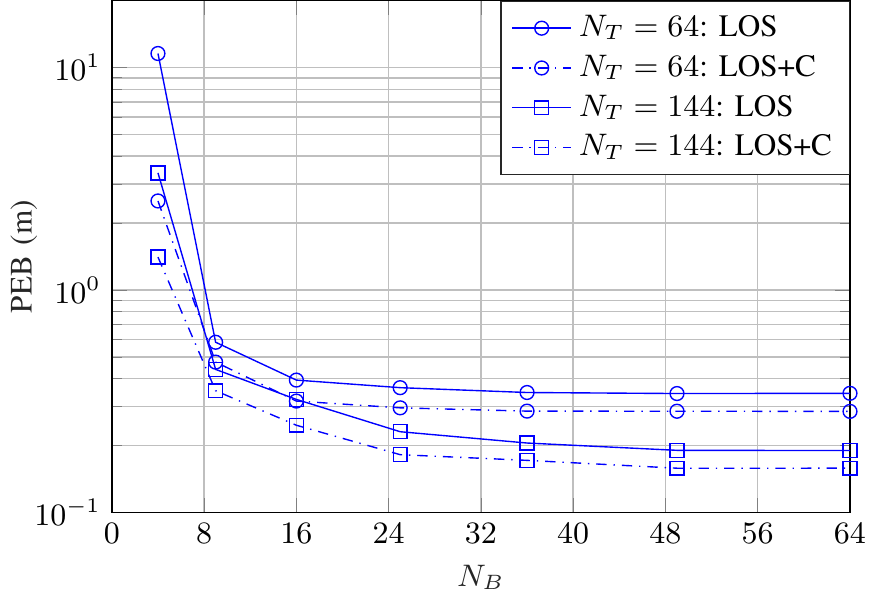}
		\vspace{-6mm}
	\caption{\textcolor{black}{Effect of $N_{\mathrm{B}}$ on the exact downlink PEB with $N_\mathrm{R}=144, N_\mathrm{T}\in\{64, 144\}$, for LOS and LOS+C, at $\mathrm{CDF}=0.9$.}}
	\label{fig:PEB_NB}
	\vspace{-9mm}
\end{figure}
In this section, we will evaluate the impact of the number of beams on downlink localization. Considering directional beamforming and a given number of transmit antennas, i.e., a fixed beamwidth, the selection of $N_{\mathrm{B}}$ becomes a trade-off between hardware complexity and the coverage area up to a certain value of $N_{\mathrm{B}}$, where more beams do not necessarily assist the localization. This relationship is highlighted in Fig.~\ref{fig:PEB_NB} for PEB values across the space, at a CDF of $90\%$ (similar results hold for the OEB, not shown). It can be seen that at a small $N_{\mathrm{B}}$, the bounds are high, but as $N_{\mathrm{B}}$ increases, the bounds start to decrease due to better coverage. However, as $N_{\mathrm{B}}$ continues to increase, the bounds reach a floor and adding more beams only adds more complexity while providing negligible improvement. \textcolor{black}{To see why, recall that the total transmitted power over the sector is fixed. So, starting with a small $N_\mathrm{B}$ and increasing it gradually improves the coverage, while reducing the power per beam. Eventually, beams start to overlap, but this does not improve the performance because the power impinging on a certain area remains approximately constant.} \textcolor{black}{This means that $N_\mathrm{B}$ should be selected as a function of the beamwidth, which is in turn a function of $N_\mathrm{T}$. For instance, considering $N_\mathrm{T} \in \{64, 144\}$,  beams with $N_\mathrm{T}=64$ are wider\footnote{\textcolor{black}{From \cite{VanTrees2002}, a URA is considered as two ULAs in orthogonal directions. The half-power beamwidth in each direction is given by $\mathrm{HPBW}=2\sin^{-1}\left(\frac{0.891}{N_\mathrm{T}}\right)$. Thus, high $N_\mathrm{T}$ leads to small $\mathrm{HPBW}$.}} than $N_\mathrm{T}=144$. Therefore, smaller $N_\mathrm{B}$ is required to provide full area coverage when  $N_\mathrm{T} =64$. More specifically, it is sufficient to have $16$ beams when $N_\mathrm{T}=64$, compared to $25$  beams in the case of $N_\mathrm{T}=144$. However, it should be noted that while a higher $N_\mathrm{T}$ provides narrower beams and necessitates more beams for coverage, it provides higher array gain, i.e., higher SNR due to scaling with $\gamma$  in \eqref{eq:bigFIMS}, hence, lower PEB and OEB. This conclusion manifests in Fig. \ref{fig:PEB_NB}, in that the use of 144 antennas attains a lower floor than 64 antennas.}

\subsection{Downlink vs. Uplink Comparison}
We now compare uplink and downlink in terms of: (i) UE orientation; (ii) number of transmit antennas; (iii) number of receive antennas. We recall that for the downlink transmitter (BS) has a known position and orientation, while for the uplink transmitter (UE) they are unknown.

\subsubsection{UE Orientation impact on PEB and OEB}
\begin{figure}[!t]
	\centering
	\footnotesize
	\vspace{-3mm}
	\includegraphics{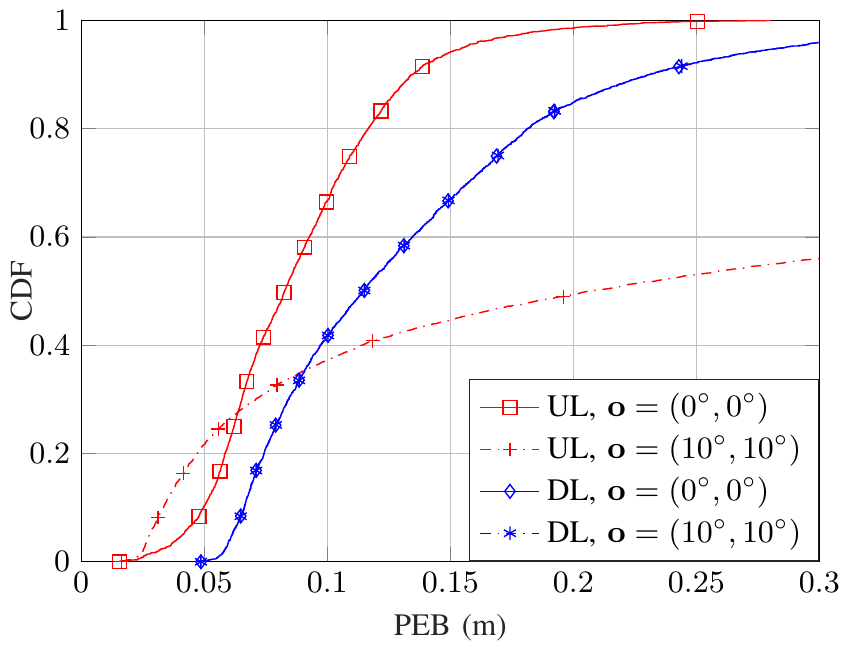}
	\vspace{-4mm}
	\caption{CDF of the PEB over the entire sector, for uplink and downlink, with different orientation angles. }\label{fig:PEB_baseline}
			\vspace{-2mm}
\end{figure}
\begin{figure}[!t]
	\centering
	\footnotesize
	\includegraphics{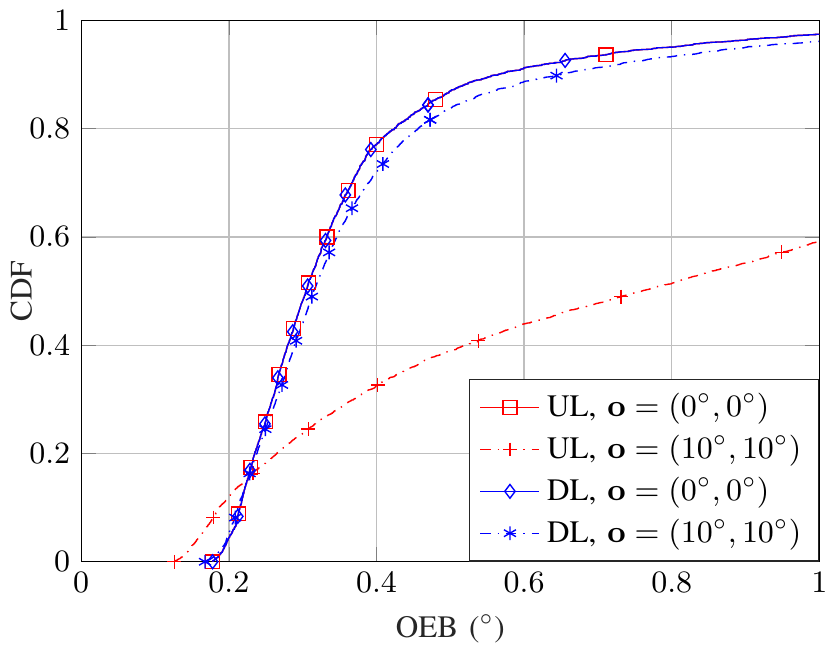}
	\vspace{-2mm}
	\caption{CDF of the OEB over the entire sector, for uplink and downlink, with different orientation angles.}
	\label{fig:OEB_baseline}
	\vspace{-7mm}
\end{figure}

Considering Fig.~\ref{fig:PEB_baseline}, the CDF of PEB is shown for uplink and downlink with two different UE orientation angles. \textcolor{black}{Recall that in the downlink, the UE is a receiver, where no beamforming is assumed. In that sense, and for the sake of computing the bounds regardless of the processing at the receiver, we assume that the receiver is equipped with isotropic antenna elements. Therefore, }the \emph{downlink PEB} is independent of the UE orientation, the downlink PEB is identical in both $0^\circ$ and $10^\circ$ orientation cases. On the contrary, the \emph{uplink PEB} is highly dependent on the UE orientation: beamforming from the UE is performed in fixed directions in the UE reference frame. Depending on the UE location, beams may miss the BS. With $10^\circ$ orientation, this happens more frequently, thus degrading the PEB. Finally, although in Fig.~\ref{fig:PEB_baseline} the uplink with $0^\circ$ orientation outperforms the downlink, this is not generally the case, since this depends on the choice of $N_{\mathrm{R}}$, as demonstrated below.

For the OEB in Fig.~\ref{fig:OEB_baseline}, downlink curves again coincide, with the uplink OEB for $0^\circ$ yielding similar performance. This is due to OEB being a function of DOA and DOD, which are interchangeable when UE and BS have the same orientation. However, when the UE orientation is $10^\circ$, OEB is again degraded, similar to the PEB. \textcolor{black}{Note that to improve the presentation, Figs.~\ref{fig:PEB_baseline} and \ref{fig:OEB_baseline} are truncated to show the relevant values of PEB and OEB, respectively.}


\subsubsection{Effect of $N_{\mathrm{R}}$ and $N_{\mathrm{T}}$}

Fig.~\ref{fig:PEB_NR}  shows the scaling effect of the PEB at $90\%$ CDF for LOS, which in line with Table \ref{tab:scale}, implies that uplink and downlink have different scaling exponents. This leads the two lines to cross at some value. So, regarding PEB, choosing $N_{\mathrm{R}}$ on either side of this crossing point dictates the outperforming scheme, uplink or downlink. Specifically, for very large number of receive antennas, uplink PEB becomes far better than downlink PEB. \textcolor{black}{With reference to Table \ref{tab:scale} and (\ref{eq:SOEB_SPEB_3D}), downlink $\text{PEB}\propto\frac{1}{\sqrt{{N_\mathrm{R}}}}$, while uplink $\text{PEB}\propto\sqrt{\frac{c_1}{N_\mathrm{R}}+\frac{c_2}{N^2_\mathrm{R}}}$, for some constants $c_1$, and $c_2$ that depend on location, bandwidth, and path gain. For the uplink case, the first term corresponds to CRLB(TOA), while the second term corresponds to CRLB(DOA), from Fig.~\ref{fig:PEB_NR}, it can be inferred that CRLB(TOA) is much smaller than CRLB(DOA) yielding uplink $\text{PEB}\propto\frac{1}{N_\mathrm{R}}$, which decays faster than $\text{PEB}\propto\frac{1}{\sqrt{{N_\mathrm{R}}}}$. This also means that the estimation of the UE location is limited by the estimation of the angles rather than the range. }

From Table \ref{tab:scale}, however, the OEB scaling is different than PEB. This is confirmed by the results of the OEB at $90\%$ CDF for LOS shown in Fig.~\ref{fig:OEB_NR}. It can be seen that for relatively large $N_{\mathrm{R}}$, OEB scales of $1/\sqrt{N_{\mathrm{R}}}$, while for small $N_{\mathrm{R}}$, it scales of $1/N_{\mathrm{R}}$, in both uplink and downlink.

\begin{figure}[!t]
	\centering
	\footnotesize
		\vspace{-4mm}
	\includegraphics{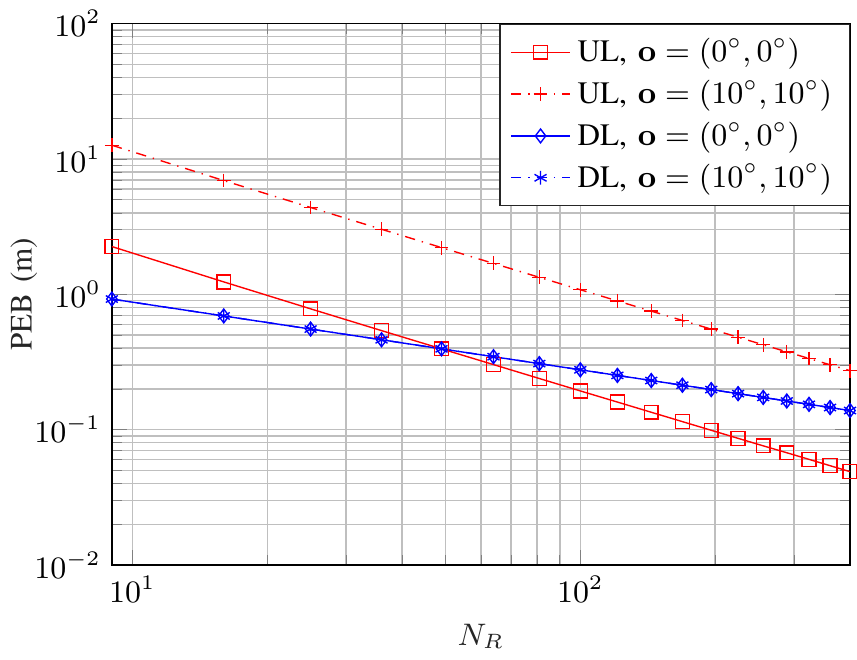}
		\vspace{-5mm}
	\caption{Scaling of the PEB w.r.t $N_{\mathrm{R}}$ for uplink and downlink LOS scenarios, at $\mathrm{CDF}=0.9$, with different orientation angles.}\label{fig:PEB_NR}
	\vspace{-4mm}
\end{figure}
\begin{figure}[!t]
	\centering
	\footnotesize
	\includegraphics{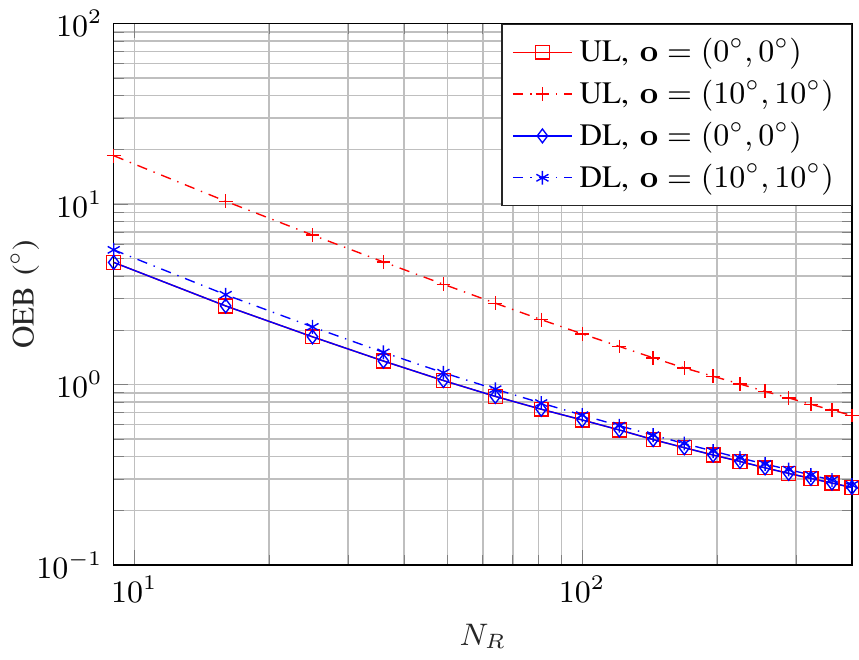}
	\caption{Scaling of the OEB w.r.t $N_{\mathrm{R}}$ for uplink and downlink LOS scenarios, at $\mathrm{CDF}=0.9$, with different orientation angles.}\label{fig:OEB_NR}
	\vspace{-2mm}
\end{figure}
\begin{figure}[!t]
	\centering
	\footnotesize
	\includegraphics{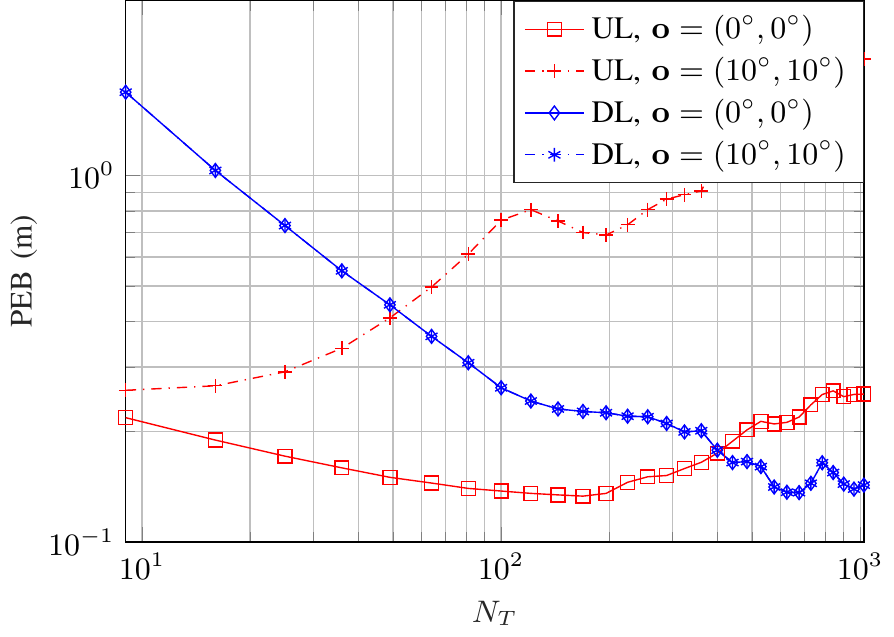}
	\vspace{-5mm}
	\caption{Scaling of the PEB w.r.t $N_\mathrm{T}$ for uplink and downlink LOS scenarios, at $\mathrm{CDF}=0.9$, with different orientation angles.}\label{fig:PEB_NT}
	\vspace{-8mm}
\end{figure}

Finally, we discuss the effect of $N_\mathrm{T}$ on the PEB shown in Fig.~\ref{fig:PEB_NT} (similar OEB results are observed, not shown). Both PEB and OEB scale non-linearly with $N_\mathrm{T}$. Small $N_\mathrm{T}$ results in bad performance due to less spatial resolution and lower SNR. As $N_\mathrm{T}$ increases, the SNR increases but the beamwidth decreases. At a certain point, the beams become too narrow and bounds start to worsen. Both uplink and downlink suffer from this effect, but it is more severe in the uplink.
\subsection{\textcolor{black}{Summary of Results}}\label{sec:summary}
\textcolor{black}{Focusing on outdoor scenarios, our simulations of the approximate approach under mmWave assumptions showed that when $N_\mathrm{R}=100$ and $N_\mathrm{T}=16$, the paths are resolvable in the space domains, and the multipath components can be considered orthogonal. Moreover, our investigations implied that NLOS clusters improve the localization when a LOS path exists. Particularly, we observed that reflectors provide modest PEB improvement for some locations, while scatterers provide small PEB decrease for more locations. Our analysis of the impact of the $N_\mathrm{B}$ and $N_\mathrm{T}$ showed that although smaller $N_\mathrm{T}$ provide better coverage due to the wider beams, larger $N_\mathrm{T}$ provides higher SNR, leading to lower PEB and OEB. We also observed that PEB and OEB are more sensitive to the orientation angle in the uplink than in the downlink. Finally, we showed that under mmWave assumptions, $\mathrm{PEB}<1$ meters and $\mathrm{OEB }<1^\circ$ are feasible for BS-UE separation of up to 50 m.}
\section{Conclusions}\label{sec:conc}
In this paper, we considered mmWave localization performance limits in terms of PEB and OEB, for uplink and downlink localization with arbitrary array geometries in multipath environments. We obtained these bounds by transforming the FIM of the channel parameters that was shown to be composed of three factors related to the receiver side, the transmitter side, and the transmitted signals. Our investigations of an approximate approach under mmWave assumptions showed that if the number of receive antennas is very high, or if the bandwidth is very large, the paths are resolvable in either the time or space domains, and the multipath components can be considered orthogonal. Consequently, the total FIM is the sum of the FIM of individual paths. We also derived closed-form expressions for single-path PEB and OEB, and showed that OEB is a function of the CRLB of the DOA and DOD, while PEB is a function of the CRLB of the TOA and the CRLB of the BS angles (DOD in the downlink, and DOA in the uplink). We showed that the uplink and downlink are not identical. Finally, although having many receive antennas is more beneficial in uplink than in downlink localization, the former is generally harder since transmit beamforming at UE may point towards directions not useful for localization.

\appendices
\section{Exact Channel Parameters FIM Submatrices}\label{app:FIM_entries}
\textcolor{black}{Taking the derivative of \eqref{eq:multipath_mu} with respect to the unknown parameters, we can write
	\small
\begin{subequations}
	\begin{align}
\frac{\partial\boldsymbol{\mu}_{\boldsymbol\varphi}}{\partial{\theta_{\mathrm{R},m}}}&=-j\sqrt{N_\mathrm{R}N_\mathrm{T}E_\mathrm{s}}\beta_m\tilde{\mathbf{K}}_{\mathrm{R},m}\mathbf{a}_{\mathrm{R},m}\mathbf{a}^\mathrm{H}_{\mathrm{T},m}\mathbf{Fs}(t-\tau_m),\\
\frac{\partial\boldsymbol{\mu}_{\boldsymbol\varphi}}{\partial{\theta_{\mathrm{T},m}}}&=j\sqrt{N_\mathrm{R}N_\mathrm{T}E_\mathrm{s}}\beta_m\mathbf{a}_{\mathrm{R},m}\mathbf{a}^\mathrm{H}_{\mathrm{T},m}\tilde{\mathbf{K}}_{\mathrm{T},m}\mathbf{Fs}(t-\tau_m),\\
\frac{\partial\boldsymbol{\mu}_{\boldsymbol\varphi}}{\partial{\phi_{\mathrm{R},m}}}&=-j\sqrt{N_\mathrm{R}N_\mathrm{T}E_\mathrm{s}}\beta_m\tilde{\mathbf{P}}_{\mathrm{R},m}\mathbf{a}_{\mathrm{R},m}\mathbf{a}^\mathrm{H}_{\mathrm{T},m}\mathbf{Fs}(t-\tau_m),\\
\frac{\partial\boldsymbol{\mu}_{\boldsymbol\varphi}}{\partial{\phi_{\mathrm{T},m}}}&=j\sqrt{N_\mathrm{R}N_\mathrm{T}E_\mathrm{s}}\beta_m\mathbf{a}_{\mathrm{R},m}\mathbf{a}^\mathrm{H}_{\mathrm{T},m}\tilde{\mathbf{P}}_{\mathrm{T},m}\mathbf{Fs}(t-\tau_m),\\
\frac{\partial\boldsymbol{\mu}_{\boldsymbol\varphi}}{\partial{\beta_{\mathrm{R},m}}}&=-j\frac{\partial\boldsymbol{\mu}_{\boldsymbol\varphi}}{\partial{\beta_{\mathrm{I},m}}}=\sqrt{N_\mathrm{R}N_\mathrm{T}E_\mathrm{s}}\mathbf{a}_{\mathrm{R},m}\mathbf{a}^\mathrm{H}_{\mathrm{T},m}\mathbf{Fs}(t-\tau_m),\\
\frac{\partial\boldsymbol{\mu}_{\boldsymbol\varphi}}{\partial{\tau_{m}}}&=\sqrt{N_\mathrm{R}N_\mathrm{T}E_\mathrm{s}}\beta_m\mathbf{a}_{\mathrm{R},m}\mathbf{a}^\mathrm{H}_{\mathrm{T},m}\mathbf{F}\frac{\partial\mathbf{s}(t-\tau_m)}{\partial{\tau_{m}}},
\end{align}
\end{subequations}
\normalsize
where $\tilde{\mathbf{K}}_{\mathrm{R},m}$ and $\tilde{\mathbf{P}}_{\mathrm{R},m}$ are as defined in \eqref{eq:RXmatrices}. $\tilde{\mathbf{K}}_{\mathrm{T},m}$ and $\tilde{\mathbf{P}}_{\mathrm{T},m}$ are defined similarly by replacing the subscript $\mathrm{R}$ by $\mathrm{T}$. Defining $\gamma\triangleq{N_\mathrm{R}N_\mathrm{T}N_\mathrm{s}E_\mathrm{s}/N_0}$, then for $1\leq{u,v}\leq{M}$, from \eqref{eq:FIM} we write
	\small
\begin{align}
[\mathbf{J}_{{\theta}_\mathrm{R}{\theta}_\mathrm{R}}]_{u,v}&=\frac{1}{N_0}\int_0^{T_o}\Re\left\lbrace\frac{\partial\boldsymbol{\mu}^\mathrm{H}_{\boldsymbol{\varphi}}}{\partial{\theta_{R,u}}} \frac{\partial\boldsymbol{\mu}_{\boldsymbol\varphi}}{\partial{\theta_{R,v}}} \right\rbrace dt 
=\gamma\Re\left\lbrace \beta_u^{*}\beta_v[\mathbf{R}_0]_{u,v}\mathbf{a}^\mathrm{H}_{R,u}\tilde{\mathbf{K}}_{R,u}\tilde{\mathbf{K}}_{R,v}\mathbf{a}_{R,v}\mathbf{a}^\mathrm{H}_{T,v}\mathbf{F}\mathbf{F}^\mathrm{H}\mathbf{a}_{T,u}\right\rbrace \label{eq:14}
\end{align}
\normalsize
where $\int_{0}^{T_o}\mathbf{s}(t-\tau_v)\mathbf{s}^\mathrm{H}(t-\tau_u)=N_\mathrm{s}[\mathbf{R}_0]_{u,v}\mathbf{I}_{N_B}$, and 
\begin{align}
[\mathbf{R}_0]_{u,v}\triangleq\int_0^{T_o}{s}_\ell(t-\tau_v){s}_\ell^*(t-\tau_u)\mathrm{d}t&=\int_{-W/2}^{W/2}|P(f)|^2e^{-j2\pi f\Delta\tau_{uv}}\mathrm{d}f\label{eq:R_0},
\end{align}
where $\Delta\tau_{uv}=\tau_v-\tau_u$. Note that in (\ref{eq:14}), the fact that $\mathbf{a}^\mathrm{H}\mathbf{b}\mathbf{c}^\mathrm{H}\mathbf{d}=\mathbf{c}^\mathrm{H}\mathbf{d}\mathbf{a}^\mathrm{H}\mathbf{b}$ is used, and that (\ref{eq:R_0}) follows from Parseval's theorem. Rearranging (\ref{eq:14}) to
	\begin{align}
	[\mathbf{J}_{{\theta}_\mathrm{R}{\theta}_\mathrm{R}}]_{u,v}=\gamma\Re\left\lbrace (\beta_u^{*}\mathbf{a}^\mathrm{H}_{R,u}\tilde{\mathbf{K}}_{R,u}\tilde{\mathbf{K}}_{R,v}\mathbf{a}_{R,v}\beta_v)(\mathbf{a}^\mathrm{H}_{T,v}\mathbf{F}\mathbf{F}^\mathrm{H}\mathbf{a}_{T,u})[\mathbf{R}_0]_{u,v}\right\rbrace,
	\end{align}
then, denoting the Hadamard product by $\odot$ and using the notation in \eqref{eq:RXmatrices}, we rewrite $\mathbf{J}_{{\theta}_\mathrm{R}{\theta}_\mathrm{R}}$
\begin{align}
\mathbf{J}_{{\theta}_\mathrm{R}{\theta}_\mathrm{R}}=\Re \left\lbrace(\mathbf{B}^\mathrm{H}{\mathbf{K}}_\mathrm{R}^\mathrm{H}{\mathbf{K}}_\mathrm{R}\mathbf{B})\odot(\mathbf{A}_\mathrm{T}^\mathrm{H}\mathbf{F}\mathbf{F}^\mathrm{H}\mathbf{A}_\mathrm{T})^\mathrm{T}\odot\mathbf{R}_0\right\rbrace.
\end{align}}
The other sub-matrices of (\ref{eq:matrix_I}) can be similarly obtained as
\small
\begin{subequations}\label{eq:bigFIMS}
	\begin{align}
	&\mathbf{J}_{\boldsymbol\theta_{\mathrm{T}}\boldsymbol\theta_{\mathrm{T}}}=\gamma  \Re \{(\mathbf{B}^\mathrm{H}\mathbf{A}_{\mathrm{R}}^{\mathrm{H}}\mathbf{A}_{\mathrm{R}}\mathbf{B})\odot({\mathbf{K}}_\mathrm{T}^{\mathrm{H}}\mathbf{FF}^\mathrm{H}{\mathbf{K}}_\mathrm{T})^\mathrm{T}\odot\mathbf{R}_0\}\label{eq:thTthT},\\
	&\mathbf{J}_{\boldsymbol\phi_{\mathrm{R}}\boldsymbol\phi_{\mathrm{R}}}=\gamma  \Re \{(\mathbf{B}^\mathrm{H}{\mathbf{P}}_{\mathrm{R}}^{\mathrm{H}}{\mathbf{P}}_{\mathrm{R}}\mathbf{B})\odot(\mathbf{A}_\mathrm{T}^{\mathrm{H}}\mathbf{F}\mathbf{F}^\mathrm{H}\mathbf{A}_\mathrm{T})^\mathrm{T}\odot\mathbf{R}_0\},\\
	&\mathbf{J}_{\boldsymbol\phi_{\mathrm{T}}\boldsymbol\phi_{\mathrm{T}}}=\gamma  \Re \{(\mathbf{B}^\mathrm{H}\mathbf{A}_{\mathrm{R}}^{\mathrm{H}}\mathbf{A}_{\mathrm{R}}\mathbf{B})\odot({\mathbf{P}}_\mathrm{T}^{\mathrm{H}}\mathbf{FF}^\mathrm{H}{\mathbf{P}}_\mathrm{T})^\mathrm{T}\odot\mathbf{R}_0\},\\
	&\mathbf{J}_{\boldsymbol\beta_{\mathrm{R}}\boldsymbol\beta_{\mathrm{R}}}=\mathbf{J}_{\boldsymbol\beta_{\mathrm{I}}\boldsymbol\beta_{\mathrm{I}}}=\Re \{(\mathbf{A}_{\mathrm{R}}^{\mathrm{H}}\mathbf{A}_{\mathrm{R}})\odot(\mathbf{A}_\mathrm{T}^{\mathrm{H}}\mathbf{F}\mathbf{F}^\mathrm{H}\mathbf{A}_\mathrm{T})^\mathrm{T}\odot\mathbf{R}_0\},\\
	&\mathbf{J}_{\boldsymbol\tau\boldsymbol\tau}=\gamma \Re \{(\mathbf{B}^\mathrm{H}\mathbf{A}_{\mathrm{R}}^{\mathrm{H}}\mathbf{A}_{\mathrm{R}}\mathbf{B})\odot(\mathbf{A}_\mathrm{T}^{\mathrm{H}}\mathbf{FF}^\mathrm{H}\mathbf{A}_\mathrm{T})^\mathrm{T}\odot\mathbf{R}_2\},\\
	&\mathbf{J}_{\boldsymbol\theta_{\mathrm{R}}\boldsymbol\theta_{\mathrm{T}}}=\gamma\Im \{ j(\mathbf{B}^\mathrm{H}{\mathbf{K}_{\mathrm{R}}}^\mathrm{H}\mathbf{A}_{\mathrm{R}}\mathbf{B})\odot({\mathbf{K}}_\mathrm{T}^{\mathrm{H}}\mathbf{FF}^\mathrm{H}\mathbf{A}_\mathrm{T})^\mathrm{T}\odot\mathbf{R}_0\},\\
	&\mathbf{J}_{\boldsymbol\theta_{\mathrm{R}}\boldsymbol\phi_{\mathrm{R}}}=\gamma  \Re \{(\mathbf{B}^\mathrm{H}{\mathbf{K}}_{\mathrm{R}}^{\mathrm{H}}{\mathbf{P}}_{\mathrm{R}}\mathbf{B})\odot(\mathbf{A}_\mathrm{T}^{\mathrm{H}}\mathbf{F}\mathbf{F}^\mathrm{H}\mathbf{A}_\mathrm{T})^\mathrm{T}\odot\mathbf{R}_0\},\\
	&\mathbf{J}_{\boldsymbol\theta_{\mathrm{R}}\boldsymbol\phi_{\mathrm{T}}}=\gamma\Im \{ j(\mathbf{B}^\mathrm{H}{\mathbf{K}_{\mathrm{R}}}^\mathrm{H}\mathbf{A}_{\mathrm{R}}\mathbf{B})\odot({\mathbf{P}}_\mathrm{T}^{\mathrm{H}}\mathbf{FF}^\mathrm{H}\mathbf{A}_\mathrm{T})^\mathrm{T}\odot\mathbf{R}_0\},\\
	&\mathbf{J}_{\boldsymbol\theta_{\mathrm{T}}\boldsymbol\phi_{\mathrm{R}}}=\gamma  \Im \{j(\mathbf{B}^\mathrm{H}\mathbf{A}_{\mathrm{R}}^{\mathrm{H}}{\mathbf{P}}_{\mathrm{R}}\mathbf{B})\odot(\mathbf{A}_\mathrm{T}^{\mathrm{H}}\mathbf{FF}^\mathrm{H}{\mathbf{K}}_\mathrm{T})^\mathrm{T}\odot\mathbf{R}_0\},\\
	&\mathbf{J}_{\boldsymbol\theta_{\mathrm{T}}\boldsymbol\phi_{\mathrm{T}}}=\gamma  \Re \{(\mathbf{B}^\mathrm{H}\mathbf{A}_{\mathrm{R}}^{\mathrm{H}}\mathbf{A}_{\mathrm{R}}\mathbf{B})\odot({\mathbf{P}}_\mathrm{T}^{\mathrm{H}}\mathbf{FF}^\mathrm{H}{\mathbf{K}}_\mathrm{T})^\mathrm{T}\odot\mathbf{R}_0\},\\
	&\mathbf{J}_{\boldsymbol\phi_{\mathrm{R}}\boldsymbol\phi_{\mathrm{T}}}=\gamma  \Im \{j(\mathbf{B}^\mathrm{H}{\mathbf{P}}_{\mathrm{R}}^{\mathrm{H}}\mathbf{A}_{\mathrm{R}}\mathbf{B})\odot({\mathbf{P}}_\mathrm{T}^{\mathrm{H}}\mathbf{FF}^\mathrm{H}\mathbf{A}_\mathrm{T})^\mathrm{T}\odot\mathbf{R}_0\},\\
	&\mathbf{J}_{\boldsymbol\beta_{\mathrm{R}}\boldsymbol\beta_{\mathrm{I}}}=\gamma \Re \{ j(\mathbf{A}_{\mathrm{R}}^{\mathrm{H}}\mathbf{A}_{\mathrm{R}})\odot(\mathbf{A}_\mathrm{T}^{\mathrm{H}}\mathbf{FF}^\mathrm{H}\mathbf{A}_\mathrm{T})^\mathrm{T}\odot\mathbf{R}_0\},\\
	&\mathbf{J}_{\boldsymbol\theta_{\mathrm{R}}\boldsymbol\tau}=\gamma \Re \{ j(\mathbf{B}^\mathrm{H}{\mathbf{K}}_{\mathrm{R}}^{\mathrm{H}}\mathbf{A}_{\mathrm{R}}\mathbf{B})\odot(\mathbf{A}_\mathrm{T}^{\mathrm{H}}\mathbf{FF}^\mathrm{H}\mathbf{A}_\mathrm{T})^\mathrm{T}\odot\mathbf{R}_1\},\label{eq:thetar_Tau}\\
	&\mathbf{J}_{\boldsymbol\theta_{\mathrm{T}}\boldsymbol\tau}=\gamma \Im \{ (\mathbf{B}^\mathrm{H}\mathbf{A}_{\mathrm{R}}^{\mathrm{H}}\mathbf{A}_{\mathrm{R}}\mathbf{B})\odot(\mathbf{A}_\mathrm{T}^{\mathrm{H}}\mathbf{FF}^\mathrm{H}{\mathbf{K}}_\mathrm{T})^\mathrm{T}\odot\mathbf{R}_1\},\label{eq:thetat_Tau}\\
	&\mathbf{J}_{\boldsymbol\phi_{\mathrm{R}}\boldsymbol\tau}=\gamma \Re \{ j(\mathbf{B}^\mathrm{H}{\mathbf{P}}_{\mathrm{R}}^{\mathrm{H}}\mathbf{A}_{\mathrm{R}}\mathbf{B})\odot(\mathbf{A}_\mathrm{T}^{\mathrm{H}}\mathbf{FF}^\mathrm{H}\mathbf{A}_\mathrm{T})^\mathrm{T}\odot\mathbf{R}_1\},\label{eq:phir_Tau}\\
	&\mathbf{J}_{\boldsymbol\phi_{\mathrm{T}}\boldsymbol\tau}=\gamma \Im \{ (\mathbf{B}^\mathrm{H}\mathbf{A}_{\mathrm{R}}^{\mathrm{H}}\mathbf{A}_{\mathrm{R}}\mathbf{B})\odot(\mathbf{A}_\mathrm{T}^{\mathrm{H}}\mathbf{FF}^\mathrm{H}{\mathbf{P}}_\mathrm{T})^\mathrm{T}\odot\mathbf{R}_1\},\label{eq:phit_Tau}\\
	&\mathbf{J}_{\boldsymbol\theta_{\mathrm{R}}\boldsymbol\beta_{\mathrm{I}}}+j\mathbf{J}_{\boldsymbol\theta_{\mathrm{R}}\boldsymbol\beta_{\mathrm{R}}}=-\gamma~ (\mathbf{B}^\mathrm{H}{\mathbf{K}}_{\mathrm{R}}^{\mathrm{H}}\mathbf{A}_{\mathrm{R}})\odot(\mathbf{A}_\mathrm{T}^{\mathrm{H}}\mathbf{FF}^\mathrm{H}\mathbf{A}_\mathrm{T})^\mathrm{T}\odot\mathbf{R}_0,\\	
	&\mathbf{J}_{\boldsymbol\theta_{\mathrm{T}}\boldsymbol\beta_{\mathrm{I}}}+j\mathbf{J}_{\boldsymbol\theta_{\mathrm{T}}\boldsymbol\beta_{\mathrm{R}}}=\gamma~ (\mathbf{B}^\mathrm{H}\mathbf{A}_{\mathrm{R}}^{\mathrm{H}}\mathbf{A}_{\mathrm{R}})\odot(\mathbf{A}_\mathrm{T}^{\mathrm{H}}\mathbf{FF}^\mathrm{H}{\mathbf{K}}_\mathrm{T})^\mathrm{T}\odot\mathbf{R}_0,\\
	&\mathbf{J}_{\boldsymbol\phi_{\mathrm{R}}\boldsymbol\beta_{\mathrm{I}}}+j\mathbf{J}_{\boldsymbol\phi_{\mathrm{R}}\boldsymbol\beta_{\mathrm{R}}}=-\gamma~ (\mathbf{B}^\mathrm{H}{\mathbf{P}}_{\mathrm{R}}^{\mathrm{H}}\mathbf{A}_{\mathrm{R}})\odot(\mathbf{A}_\mathrm{T}^{\mathrm{H}}\mathbf{FF}^\mathrm{H}\mathbf{A}_\mathrm{T})^\mathrm{T}\odot\mathbf{R}_0,\\
	&\mathbf{J}_{\boldsymbol\phi_{\mathrm{T}}\boldsymbol\beta_{\mathrm{I}}}+j\mathbf{J}_{\boldsymbol\phi_{\mathrm{T}}\boldsymbol\beta_{\mathrm{R}}}=\gamma~ (\mathbf{B}^\mathrm{H}\mathbf{A}_{\mathrm{R}}^{\mathrm{H}}\mathbf{A}_{\mathrm{R}})\odot(\mathbf{A}_\mathrm{T}^{\mathrm{H}}\mathbf{FF}^\mathrm{H}{\mathbf{P}}_\mathrm{T})^\mathrm{T}\odot\mathbf{R}_0,\\
	&\mathbf{J}_{\boldsymbol\beta_{\mathrm{R}}\boldsymbol\tau}+j\mathbf{J}_{\boldsymbol\beta_{\mathrm{I}}\boldsymbol\tau}=\gamma~(\mathbf{A}^\mathrm{H}\mathbf{A}_{\mathrm{R}}\mathbf{B})\odot(\mathbf{A}_\mathrm{T}^{\mathrm{H}}\mathbf{FF}^\mathrm{H}\mathbf{A}_\mathrm{T})^\mathrm{T}\odot\mathbf{R}_1,\label{eq:betaTau}
	\end{align}
\end{subequations}
\normalsize
\textcolor{black}{where, similar to $\mathbf{R}_0$ the elements of $\mathbf{R}_1$ and $\mathbf{R}_2$ are given by
\begin{align}
[\mathbf{R}_1]_{u,v}\triangleq\int_0^{T_o}\frac{\partial{s}_\ell(t-\tau_v)}{\partial\tau_v}{s}_\ell^*(t-\tau_u)\mathrm{d}t&=\int_{-W/2}^{W/2}2\pi f|P(f)|^2e^{-j2\pi f\Delta\tau_{uv}}\mathrm{d}f\label{eq:R_1},\\
[\mathbf{R}_2]_{u,v}\triangleq\int_0^{T_o}\frac{\partial{s}_\ell(t-\tau_v)}{\partial\tau_v}\frac{\partial{s}^*_\ell(t-\tau_u)}{\partial\tau_u}\mathrm{d}t&=\int_{-W/2}^{W/2}(2\pi f)^2|P(f)|^2e^{-j2\pi f\Delta\tau_{uv}}\mathrm{d}f\label{eq:R_2}.
\end{align}}
\vspace{-6mm}
\section{CRLBs of Single-Path Parameters}\label{app:CRLB}
\textcolor{black}{Evaluating the submatrices in (\ref{eq:bigFIMS}) for the single-path case ($M=1$) then using the notation introduced in Section \ref{sec:los_crlb} to simplify presentation, it can be seen that
\begin{align}
\CRLB(\tau)&\triangleq\frac{1}{J_{\tau\tau}}=\frac{1}{4\gamma\pi^2|\beta|^2GW_{\mathrm{eff}}^2}.
\end{align}
Moreover, the FIM of $\boldsymbol{\varphi}_{\boldsymbol{\theta},\boldsymbol{\phi},\boldsymbol{\beta}}\triangleq[\theta_\mathrm{R},\theta_\mathrm{T},\phi_\mathrm{R},\phi_\mathrm{T},\beta_\mathrm{R},\beta_\mathrm{I}]^\mathrm{T}$ can be obtained as}
\begin{align}
\mathbf{J}_{\boldsymbol{\theta},\boldsymbol{\phi},\boldsymbol{\beta}}&\triangleq\gamma|\beta|^2\left[\begin{array}{c:c}
\mathbf{A}&\mathbf{U}\\
\hdashline
\mathbf{U}^\mathrm{T}&\mathbf{C}\end{array}\right]=\gamma|\beta|^2\left[\begin{array}{cccc:cc}
R_\theta G& 0 &X_{\theta ,\phi}G &0 &0 &0 \\
0& T_\theta &0&Y'_{\theta ,\phi}&S_1&U_1\\
X_{\theta ,\phi}G&0&R_\phi G&0&0&0\\
0&Y'_{\theta ,\phi}&0& T_\phi &S_2&U_2\\
\hdashline
0&S_1&0&S_2&\frac{G}{|\beta|^2}&0\\
0&U_1&0&U_2&0&\frac{G}{|\beta|^2}\\
\end{array}
\right],\label{eq:partitioned2}
\end{align}
where $S_1\triangleq-\Im\{\beta V_\theta \}/|\beta|^2, U_1\triangleq\Re\{\beta V_\theta \}/|\beta|^2$, $S_2\triangleq-\Im\{\beta V_\phi \}/|\beta|^2,$ and $U_2\triangleq\Re\{\beta V_\phi \}/|\beta|^2$. Then, we compute the CRLBs of DOA and DOD using their EFIM as
\begin{align}
\mathbf{J}^\mathrm{e}_{{\boldsymbol{\theta},\boldsymbol{\phi}}}&=\mathbf{A-UC}^{-1}\mathbf{U}^\mathrm{T}=\gamma|\beta|^2\left[\begin{array}{cccc}
R_\theta G& 0 &X_{\theta ,\phi}G &0\\
0& \frac{L_\theta }{G}&0&\frac{Y_{\theta ,\phi}}{G}\\
X_{\theta ,\phi}G&0&R_\phi G&0\\
0&\frac{Y_{\theta ,\phi}}{G}&0& \frac{L_\phi }{G}\\
\end{array}
\right]\label{eq:FIM_DOD_DOA}.
\end{align}
To simplify the inverse computation, we utilize the independence between DOA and DOD to re-order $\mathbf{J}^\mathrm{e}_{{\boldsymbol{\theta},\boldsymbol{\phi}}}$ to write (\ref{eq:crlb:expressions}), from which (\ref{eq:closed_CRLB}) follow.
\begin{subequations}\label{eq:crlb:expressions}
	\begin{align}
	\CRLB\left(\theta_\mathrm{R},\phi_\mathrm{R}\right)&=
	\frac{1}{\gamma|\beta|^2G}\begin{bmatrix}
	R_\theta&X_{\theta ,\phi}\\
	X_{\theta ,\phi}&R_\phi \end{bmatrix}^{-1},\\
	\CRLB\left(\theta_\mathrm{T},\phi_\mathrm{T}\right)&=
	\frac{G}{\gamma|\beta|^2}\begin{bmatrix}
	L_\theta &Y_{\theta ,\phi}\\
	Y_{\theta ,\phi}&L_\phi \end{bmatrix}^{-1},
	\end{align}
\end{subequations}
\vspace{-6mm}
\section{Non-Zero Elements of $\boldsymbol\Upsilon$}\label{app:derivatives}
\textcolor{black}{For the LOS case, it can be shown that
	\small
	\begin{align}\label{eq:derivatives_los}
	\begin{aligned}
	\frac{\partial \phi_\mathrm{UE,1}}{\partial\phi_0}&=\frac{-p_x'^2\cos\theta_0+(p_y\cos\phi_0-p_x\sin\phi_0)p_y'}{p_x'^2+p_y'^2},\\
	\frac{\partial \phi_\mathrm{UE,1}}{\partial\theta_0}&=-\frac{p_x'p_z'}{p_x'^2+p_y'^2},\\
	\frac{\partial \phi_\mathrm{UE,1}}{\partial\mathbf{p}}&=(\mathbf{r}_2\mathbf{r}_1^\mathrm{T}-\mathbf{r}_1\mathbf{r}_2^\mathrm{T})\frac{\mathbf{p}}{p_x'^2+p_y'^2},\\
		\frac{\partial \phi_\mathrm{BS,1}}{\partial\mathbf{p}}&=\frac{[-p_y,p_x,0]^\mathrm{T}}{p_x^2+p_y^2},\\
	\frac{\partial \theta_\mathrm{BS,1}}{\partial\mathbf{p}}&=\frac{[p_xp_z,p_yp_z,-(p_x^2+p_y^2)]^\mathrm{T}}{\|\mathbf{p}\|\sqrt{p_x^2+p_y^2}},
		\end{aligned}\qquad
			\begin{aligned}
				\frac{\partial \theta_\mathrm{UE,1}}{\partial\phi_0}&=\frac{p_x'\sin\theta_0}{\sqrt{p_x'^2+p_y'^2}},\\
	\frac{\partial \theta_\mathrm{UE,1}}{\partial\theta_0}&=-\frac{p_y'}{\sqrt{p_x'^2+p_y'^2}},\\
	\frac{\partial \theta_\mathrm{UE,1}}{\partial\mathbf{p}}&=\frac{1}{\sqrt{p_x'^2+p_y'^2}}\left(\mathbf{r}_3+\frac{p_z'}{\|\mathbf{p}\|}\mathbf{p}\right).\\
\frac{\partial \tau_1}{\partial\mathbf{p}}&=\frac{\mathbf{p}}{c\|\mathbf{p}\|},
			\end{aligned}
	\end{align}
\normalsize
where $\mathbf{r}_i, 1\leq{i}\leq{3}$ is the $i^\mathrm{th}$ column of $\mathbf{R}(\theta_0,\phi_0)$. }

\textcolor{black}{For the NLOS case, we use the similarity in \eqref{LOS_geometry_BS}, \eqref{LOS_geometry_UE}, and \eqref{NLOS_geometry}, to obtain, $\frac{\partial \phi_\mathrm{UE,m}}{\partial\phi_0}$, $\frac{\partial \phi_\mathrm{UE,m}}{\partial\theta_0}$, $\frac{\partial \phi_\mathrm{UE,m}}{\partial\mathbf{p}}$, $\frac{\partial \theta_\mathrm{UE,m}}{\partial\phi_0}$, $\frac{\partial \theta_\mathrm{UE,m}}{\partial\theta_0}$, $\frac{\partial \theta_\mathrm{UE,m}}{\partial\mathbf{p}}$, $\frac{\partial \tau_m}{\partial\mathbf{p}}$ by replacing $\mathbf{p}$ and $\mathbf{p}'$ in \eqref{eq:derivatives_los} with $\mathbf{w}_m$ and $\mathbf{w}'_m$, receptively. Note that in NLOS, $\frac{\partial \phi_\mathrm{BS,m}}{\partial\mathbf{p}}=\frac{\partial \theta_\mathrm{BS,m}}{\partial\mathbf{p}}=0$. Finally, we obtain the following
\small
\begin{align}\label{eq:derivatives_nlos}
	\begin{aligned}
		\frac{\partial \phi_\mathrm{UE,m}}{\partial\mathbf{q}_m}&=-\frac{\partial \phi_\mathrm{UE,m}}{\partial\mathbf{p}},\\
	\frac{\partial \theta_\mathrm{UE,m}}{\partial\mathbf{q}_m}&=-\frac{\partial \theta_\mathrm{UE,m}}{\partial\mathbf{p}},\\
	\frac{\partial \phi_\mathrm{BS,m}}{\partial\mathbf{q}_m}&=\frac{[-q_{m,y},q_{m,x},0]^\mathrm{T}}{q_{m,x}^2+q_{m,y}^2},\\
	\end{aligned}\qquad
\begin{aligned}
	\frac{\partial \theta_\mathrm{BS,m}}{\partial\mathbf{q}_m}&=\frac{[q_{m,x}q_{m,z},q_{m,y}q_{m,z},-(q_{m,x}^2+q_{m,y}^2)]^\mathrm{T}}{\|\mathbf{q}_m\|\sqrt{q_{m,x}^2+q_{m,y}^2}},\\
\frac{\partial \tau_m}{\partial\mathbf{q}_m}&=\frac{\mathbf{q}_m}{c\|\mathbf{q}_m\|}-\frac{\mathbf{w}_m}{c\|\mathbf{w}_m\|},\\
\mathbf{w}_m&\triangleq\mathbf{p}-\mathbf{q}_m
\end{aligned}
	\end{align}
}\normalsize
\vspace{-5mm}
\section{Proof of Proposition 2}\label{app:OP_FIM}
We start by deriving the EFIM of the location parameters for the $m^\mathrm{th}$ path. Then, we show that the overall EFIM can be written as a summation of the individual EFIM. For the $m^\mathrm{th}$ path,
\begin{align}
\mathbf{J}_{\mathbf{o,p,q}_m}^{(m)}&={\boldsymbol{\Upsilon}}_m\boldsymbol{\Lambda}^\mathrm{e}_m{\boldsymbol{\Upsilon}}_m^\mathrm{T}=\begin{bmatrix}
\overline{\boldsymbol{\Upsilon}}_m\boldsymbol{\Lambda}^\mathrm{e}_m\overline{\boldsymbol{\Upsilon}}_m^\mathrm{T}&\overline{\boldsymbol{\Upsilon}}_m\boldsymbol{\Lambda}^\mathrm{e}_m\overline{\overline{\boldsymbol{\Upsilon}}}_m^\mathrm{T}\\
\overline{\overline{\boldsymbol{\Upsilon}}}_m\boldsymbol{\Lambda}^\mathrm{e}_m\overline{\boldsymbol{\Upsilon}}_m^\mathrm{T}&\overline{\overline{\boldsymbol{\Upsilon}}}\boldsymbol{\Lambda}^\mathrm{e}_m\overline{\overline{\boldsymbol{\Upsilon}}}^\mathrm{T}
\end{bmatrix}.
\end{align}
Consequently, for $\mathbf{o}$ and $\mathbf{p}$, by Schur's complement,
\begin{align}
\mathbf{J}_{\mathbf{o,p}}^{(m)}&=\overline{\boldsymbol{\Upsilon}}_m\boldsymbol{\Lambda}^\mathrm{e}_m\overline{\boldsymbol{\Upsilon}}_m^\mathrm{T}-\overline{\boldsymbol{\Upsilon}}_m\boldsymbol{\Lambda}^\mathrm{e}_m\overline{\overline{\boldsymbol{\Upsilon}}}_m^\mathrm{T}\left(\overline{\overline{\boldsymbol{\Upsilon}}}\boldsymbol{\Lambda}^\mathrm{e}_m\overline{\overline{\boldsymbol{\Upsilon}}}^\mathrm{T}\right)^{-1}\overline{\overline{\boldsymbol{\Upsilon}}}_m\boldsymbol{\Lambda}^\mathrm{e}_m\overline{\boldsymbol{\Upsilon}}_m^\mathrm{T}.\label{eq:mth_fim_op}
\end{align}
Recall that for $m=1$ the second term above is undefined. Focusing on all the $M$ paths, define
\begin{align}
\boldsymbol{\Upsilon}&\triangleq\begin{bmatrix}\overline{\boldsymbol{\Upsilon}}_1\overline{\boldsymbol{\Upsilon}}_2&\cdots&\overline{\boldsymbol{\Upsilon}}_M\\
\mathbf{0}& \overline{\overline{\boldsymbol{\Upsilon}}}_2&\cdots&\mathbf{0}\\
\vdots&\vdots&\ddots&\vdots&\\
\mathbf{0}&\mathbf{0}&\cdots& \overline{\overline{\boldsymbol{\Upsilon}}}_M\\
\end{bmatrix},\qquad\qquad
\mathbf{J}^\mathrm{e}_{\boldsymbol{\varphi}_{\mathrm{CH}}}\triangleq\begin{bmatrix}
\boldsymbol{\Lambda}^\mathrm{e}_1&\mathbf{0}&\cdots&\mathbf{0}\\
\mathbf{0}&\boldsymbol{\Lambda}^\mathrm{e}_2&\cdots&\mathbf{0}\\
\vdots&\vdots&\ddots&\vdots&\\
\mathbf{0}&\mathbf{0}&\cdots&\boldsymbol{\Lambda}^\mathrm{e}_M\\
\end{bmatrix}.
\end{align}
Then from (\ref{eq:tranformation}),
\begin{align}
\mathbf{J}_{\boldsymbol{\varphi}_{\mathrm{L}}}&=\boldsymbol{\Upsilon}\mathbf{J}^\mathrm{e}_{\boldsymbol{\varphi}_{\mathrm{CH}}}\boldsymbol{\Upsilon}^\mathrm{T}=\begin{bmatrix}
\sum_{m=1}^M{\overline{\boldsymbol{\Upsilon}}}_m\boldsymbol{\Lambda}^\mathrm{e}_m\overline{\boldsymbol{\Upsilon}}_m^\mathrm{T}&  {\overline{\boldsymbol{\Upsilon}}}_2\boldsymbol{\Lambda}^\mathrm{e}_2\overline{\overline{\boldsymbol{\Upsilon}}}_2^\mathrm{T}&\cdots&{\overline{\boldsymbol{\Upsilon}}}_M\boldsymbol{\Lambda}^\mathrm{e}_M\overline{\overline{\boldsymbol{\Upsilon}}}_M^\mathrm{T}\\
\overline{\overline{\boldsymbol{\Upsilon}}}_2\boldsymbol{\Lambda}^\mathrm{e}_2{\overline{\boldsymbol{\Upsilon}}}_2^\mathrm{T}&
\overline{\overline{\boldsymbol{\Upsilon}}}_2\boldsymbol{\Lambda}^\mathrm{e}_2\overline{\overline{\boldsymbol{\Upsilon}}}_2^\mathrm{T}&\cdots&\mathbf{0}\\
\vdots&\vdots&\ddots&\vdots&\\
\overline{\overline{\boldsymbol{\Upsilon}}}_M\boldsymbol{\Lambda}^\mathrm{e}_M{\overline{\boldsymbol{\Upsilon}}}_M^\mathrm{T}&\mathbf{0}&\cdots&
\overline{\overline{\boldsymbol{\Upsilon}}}_M\boldsymbol{\Lambda}^\mathrm{e}_M\overline{\overline{\boldsymbol{\Upsilon}}}_M^\mathrm{T}
\end{bmatrix}.\notag
\end{align}
By Schur's Complement and using (\ref{eq:mth_fim_op}), it is easy to verify that
\begin{align}
{\mathbf{J}}^{\mathrm{e}}_{\mathbf{o,p}}=&\sum_{m=1}^M{\overline{\boldsymbol{\Upsilon}}}_m\boldsymbol{\Lambda}^\mathrm{e}_m\overline{\boldsymbol{\Upsilon}}_m^\mathrm{T}-\sum_{m=2}^M{\overline{\boldsymbol{\Upsilon}}}_m\boldsymbol{\Lambda}^\mathrm{e}_m\overline{\overline{\boldsymbol{\Upsilon}}}_m^\mathrm{T}\left(\overline{\overline{\boldsymbol{\Upsilon}}}_m\boldsymbol{\Lambda}^\mathrm{e}_m\overline{\overline{\boldsymbol{\Upsilon}}}_m^\mathrm{T}\right)^{-1}\overline {\overline{\boldsymbol{\Upsilon}}}_m\boldsymbol{\Lambda}^\mathrm{e}_m{\overline{\boldsymbol{\Upsilon}}}_m^\mathrm{T},=\sum_{m=1}^{M}\mathbf{J}_{\mathbf{o,p}}^{(m)}.\label{eq:total_fim_op}
\end{align}
\vspace{-8mm}
\section{Closed-form PEB and OEB for LOS-only}\label{app:LOS_PEB_OEB}
\textcolor{black}{To find the LOS SOEB and SPEB in a closed form, we note that
\begin{align}
\mathbf{J}_{\mathbf{o},\mathbf{p}}^{-1}&=\left(\boldsymbol{\Upsilon}\boldsymbol\Lambda_1^\mathrm{e}\boldsymbol{\Upsilon}^\mathrm{T}\right)^{-1}=\left(\boldsymbol{\Upsilon}^{-1}\right)^\mathrm{T}\left(\boldsymbol\Lambda_1^\mathrm{e}\right)^{-1}\boldsymbol{\Upsilon}^{-1}=\left(\frac{\partial\boldsymbol{\varphi}^\mathrm{T}_\mathrm{L}}{\partial\boldsymbol{\varphi}_\mathrm{CH}}\right)^\mathrm{T}\left(\boldsymbol\Lambda_1^\mathrm{e}\right)^{-1}\left(\frac{\partial\boldsymbol{\varphi}^\mathrm{T}_\mathrm{L}}{\partial\boldsymbol{\varphi}_\mathrm{CH}}\right)\label{eq:proof},
\end{align}
where the rightmost term is obtained by the inverse function theorem \cite{fleming1987functions}. Consequently,
\begin{subequations}
\begin{align}
\mathrm{SPEB}&=\Tr\left\lbrace\left(\frac{\partial\mathbf{p}^\mathrm{T}}{\partial\boldsymbol{\varphi}_\mathrm{CH}}\right)^\mathrm{T}\left(\boldsymbol\Lambda_1^\mathrm{e}\right)^{-1}\left(\frac{\partial\mathbf{p}^\mathrm{T}}{\partial\boldsymbol{\varphi}_\mathrm{CH}}\right)\right\rbrace.\label{eq:proof_peb}\\
\mathrm{SOEB}&=\Tr\left\lbrace\left(\frac{\partial\mathbf{o}^\mathrm{T}}{\partial\boldsymbol{\varphi}_\mathrm{CH}}\right)^\mathrm{T}\left(\boldsymbol\Lambda_1^\mathrm{e}\right)^{-1}\left(\frac{\partial\mathbf{o}^\mathrm{T}}{\partial\boldsymbol{\varphi}_\mathrm{CH}}\right)\right\rbrace.\label{eq:SOEB_trace}
\end{align}	
\end{subequations}
Dropping the LOS subscript ``1" and using spherical coordinates, we write
\begin{align}
\mathbf{p}=c\tau\begin{bmatrix}\cos\phi_\mathrm{BS}\sin\theta_\mathrm{BS},\sin\phi_\mathrm{BS}\sin\theta_\mathrm{BS},\cos\theta_\mathrm{BS}\end{bmatrix}^\mathrm{T}.\label{eq:p_spherical}
\end{align}
For the uplink, $\boldsymbol{\varphi}_\mathrm{CH}=[\theta_\mathrm{BS},\theta_\mathrm{UE},\phi_\mathrm{BS},\phi_\mathrm{UE},\tau]^\mathrm{T}$ (Section \ref{sec:general_formulation}). Therefore, it follows that
\begin{align}
\left(\frac{\partial\mathbf{p}^\mathrm{T}}{\partial\boldsymbol{\varphi}_\mathrm{CH}}\right)^\mathrm{T}=c\begin{bmatrix}
\tau\cos\phi_\mathrm{BS}\cos\theta_\mathrm{BS}&0&-\tau\sin\phi_\mathrm{BS}\sin\theta_\mathrm{BS}&0&\cos\phi_\mathrm{BS}\sin\theta_\mathrm{BS}\\
\tau\sin\phi_\mathrm{BS}\cos\theta_\mathrm{BS}&0&\tau\cos\phi_\mathrm{BS}\sin\theta_\mathrm{BS}&0&\sin\phi_\mathrm{BS}\sin\theta_\mathrm{BS}\\
-\tau\sin\theta_\mathrm{BS}&0&0&0&\cos\theta_\mathrm{BS}
\end{bmatrix}\label{eq:grad_p}
\end{align}
Defining $\sigma^2_{xx}$ and $\sigma^2_{xy}, x,y\in\{\theta_\mathrm{R},\theta_\mathrm{T},\phi_\mathrm{R},\phi_\mathrm{T},\tau\}$, as respectively the CRLB of $x$, and the covariance of $x$ and $y$, then from Appendix \ref{app:CRLB}, we can write the uplink EFIM as
\begin{align}
\left({\boldsymbol\Lambda_1^\mathrm{e}}\right)^{-1}=\begin{bmatrix}
\sigma^2_{\theta_\mathrm{BS}\theta_\mathrm{BS}}&0&\sigma^2_{\theta_\mathrm{BS}\phi_\mathrm{BS}}&0&0\\
0&\sigma^2_{\theta_\mathrm{UE}\theta_\mathrm{UE}}&0&\sigma^2_{\theta_\mathrm{UE}\phi_\mathrm{UE}}&0\\
\sigma^2_{\theta_\mathrm{BS}\phi_\mathrm{BS}}&0&\sigma^2_{\phi_\mathrm{BS}\phi_\mathrm{BS}}&0&0\\
0&\sigma^2_{\theta_\mathrm{UE}\phi_\mathrm{UE}}&0&\sigma^2_{\phi_\mathrm{UE}\phi_\mathrm{UE}}&0\\
0&0&0&0&\sigma^2_{\tau\tau}
\end{bmatrix}.\label{eq:inv_efim_los}
\end{align}	
Substituting in \eqref{eq:proof_peb}, and simplifying the results yield
\begin{align}
\mathrm{SPEB}=\|\mathbf{p}\|^2\sigma^2_{\theta_\mathrm{BS}\theta_\mathrm{BS}}+\|\mathbf{p}\|^2\sin^2\theta_\mathrm{BS}\sigma^2_{\phi_\mathrm{BS}\phi_\mathrm{BS}}+c^2\sigma^2_{\tau\tau}.\label{eq:SPEB_proof}
\end{align}	
To obtain downlink SPEB, we need to exchange columns 1 and 3 with 2 and 4 in \eqref{eq:grad_p}, respectively, while concurrently swapping the role of BS and UE angles in \eqref{eq:inv_efim_los}. Eventually, this leads to the same SPEB expression in \eqref{eq:SPEB_proof}.}

\textcolor{black}{We now focus on the SOEB. From the properties of spherical coordinates,
\begin{align}
\cos\phi_{\mathrm{UE}}&=\frac{p_x'}{\|\mathbf{p}\|\sin{\theta_\mathrm{UE}}}=-\frac{\mathbf{r}_1^\mathrm{T}\mathbf{p}}{\|\mathbf{p}\|\sin{\theta_\mathrm{UE}}}\quad \Longrightarrow\quad \cos(\phi_0-\phi_{\mathrm{BS}})=-\frac{\cos\phi_{\mathrm{UE}}\sin\theta_\mathrm{UE}}{\sin\theta_{\mathrm{BS}}}.\label{eq:cos_delta_phi}
\end{align}
Differentiating both sides w.r.t to $\theta_\mathrm{BS}$, we have
\begin{align}
-\sin(\phi_0-\phi_{\mathrm{BS}})\frac{\partial\phi_0}{\partial\theta_\mathrm{BS}}=-\frac{\cos\phi_{\mathrm{UE}}\sin\theta_\mathrm{UE}}{\sin^2\theta_{\mathrm{BS}}}\cos\theta_{\mathrm{BS}}.
\end{align}
Consequently, using \eqref{eq:cos_delta_phi} we obtain,
\begin{align}
\frac{\partial\phi_0}{\partial\theta_\mathrm{BS}}&=\cot\theta_\mathrm{BS}\cot(\phi_0-\phi_{\mathrm{BS}}).
\end{align}
Similarly, it can be shown that,
\begin{align}
\frac{\partial\phi_0}{\partial\theta_\mathrm{UE}}=\frac{\cos\phi_{\mathrm{UE}}\cos\theta_\mathrm{UE}}{\sin\theta_{\mathrm{BS}}\sin(\phi_0-\phi_{\mathrm{BS}})},\qquad \frac{\partial\phi_0}{\partial\phi_\mathrm{BS}}=1,\qquad\frac{\partial\phi_0}{\partial\phi_\mathrm{UE}}=-\frac{\sin\phi_{\mathrm{UE}}\sin\theta_\mathrm{UE}}{\sin\theta_{\mathrm{BS}}\sin(\phi_0-\phi_{\mathrm{BS}})}.
\end{align}}

\textcolor{black}{On the other hand, we have
\begin{align}
\cos\theta_{\mathrm{UE}}=\frac{-\mathbf{r}_3^\mathrm{T}\mathbf{p}}{\|\mathbf{p}\|}=\sin\theta_0\sin\theta_{\mathrm{BS}}\sin(\phi_0-\phi_{\mathrm{BS}})-\cos\theta_0\cos\theta_{\mathrm{BS}},
\end{align}
Similarly, differentiating and simplifying the results yield
\begin{subequations}
\begin{align}
\frac{\partial\theta_0}{\partial\theta_\mathrm{BS}}&=\frac{\partial\phi_0}{\partial\theta_\mathrm{BS}}\frac{\partial\theta_0}{\partial\phi_\mathrm{BS}}-\frac{\tan(\theta_0)\sin(\phi_0-\phi_{\mathrm{BS}})+\tan(\theta_\mathrm{BS})}{\tan(\theta_0)+\tan(\theta_\mathrm{BS})\sin(\phi_0-\phi_{\mathrm{BS}})},\\
\frac{\partial\theta_0}{\partial\theta_\mathrm{UE}}&=-\frac{\sin(\theta_\mathrm{UE})+\sin(\theta_0)\sin(\theta_\mathrm{BS})\cos(\phi_0-\phi_{\mathrm{BS}})\frac{\partial\phi_0}{\partial\theta_\mathrm{UE}}}{\sin(\theta_\mathrm{BS})\cos(\theta_0)\sin(\phi_0-\phi_{\mathrm{BS}})+\sin(\theta_0)\cos(\theta_\mathrm{BS})},\\
\frac{\partial\theta_0}{\partial\phi_\mathrm{BS}}&=-\frac{\tan(\theta_\mathrm{BS})\tan(\theta_0)\cos(\phi_0-\phi_{\mathrm{BS}})}{\tan(\theta_0)+\tan(\theta_\mathrm{BS})\sin(\phi_0-\phi_{\mathrm{BS}})},\\
\frac{\partial\theta_0}{\partial\phi_\mathrm{UE}}&=\frac{\partial\theta_0}{\partial\phi_\mathrm{BS}}\frac{\partial\phi_0}{\partial\phi_\mathrm{UE}}.
\end{align}
\end{subequations}
Note that both $\frac{\partial\theta_0}{\partial\tau}$ and $\frac{\partial\phi_0}{\partial\tau}$ are zeros. Therefore, SOEB is a weighted sum of the angular bounds. Recalling that $\partial\mathbf{o}=[\partial\theta_0,\partial\phi_0]^\mathrm{T}$, SOEB can be written from \eqref{eq:SOEB_trace} in the form
\begin{align}
\mathrm{SOEB}=&\left\|\frac{\partial\mathbf{o}}{\partial\theta_\mathrm{BS}}\right\|^2\sigma^2_{\theta_\mathrm{BS}\theta_\mathrm{BS}}+\left\|\frac{\partial\mathbf{o}}{\partial\theta_\mathrm{UE}}\right\|^2\sigma^2_{\theta_\mathrm{UE}\theta_\mathrm{UE}}+\left\|\frac{\partial\mathbf{o}}{\partial\phi_\mathrm{BS}}\right\|^2\sigma^2_{\phi_\mathrm{BS}\phi_\mathrm{BS}}+\left\|\frac{\partial\mathbf{o}}{\partial\phi_\mathrm{UE}}\right\|^2\sigma^2_{\phi_\mathrm{UE}\phi_\mathrm{UE}}\notag\\&+2\left(\frac{\partial\mathbf{o}^\mathrm{T}}{\partial\theta_\mathrm{BS}}\frac{\partial\mathbf{o}}{\partial\phi_\mathrm{BS}}\right)\sigma^2_{\theta_\mathrm{BS}\phi_\mathrm{BS}}+2\left(\frac{\partial\mathbf{o}^\mathrm{T}}{\partial\theta_\mathrm{UE}}\frac{\partial\mathbf{o}}{\partial\phi_\mathrm{UE}}\right)\sigma^2_{\theta_\mathrm{UE}\phi_\mathrm{UE}}.\label{eq:SOEB_proof}
\end{align}
Note that applying the column swapping procedure, described after \eqref{eq:SPEB_proof}, to obtain the downlink SPEB leads to the same expression as \eqref{eq:SOEB_proof}, hence \eqref{eq:SOEB_proof} is valid for both the uplink and downlink.}
\vspace{-10mm}
\bibliographystyle{IEEEtran}
\bibliography{bib_paper2}
\end{document}